\documentclass[onecolumn,draftcls]{IEEEtran}
\usepackage{epsfig}
\usepackage{times}
\usepackage{float}
\usepackage{afterpage}
\usepackage{amsmath}
\usepackage{mathrsfs}
\usepackage{amstext}
\usepackage{amssymb,bm}
\usepackage{latexsym}
\usepackage{color}
\usepackage{graphicx}
\usepackage{amsmath}
\usepackage{amsthm}
\usepackage{tikz}
\usepackage{tikzscale}
\usepackage{graphicx}
\usepackage[center]{caption}
\usepackage{pstricks}
\usepackage{subfigure}
\usepackage{booktabs}
\usepackage{multicol}
\usepackage{lipsum}
\usepackage{dblfloatfix}
\usepackage{algorithm} 
\usepackage[noend]{algpseudocode} 
\usepackage[normalem]{ulem}
\usepackage{enumerate}
\usepackage{bbm}
\usepackage{dsfont}
\usepackage{tikz}

\newcommand{\pr}{\mathbb{P}}
\newcommand{\E}{\mathbb{E}}
\newcommand{\1}{\mathds{1}}

\newcommand{\mbb}{\mathbb}
\newcommand{\mcal}{\mathcal}
\newcommand{\msf}{\mathsf}

\newtheorem{thm}{Theorem}

\newtheorem{lem}[thm]{Lemma}

\newtheorem{rem}{Remark}

\newtheorem*{defin*}{Definition}
\newtheorem{prope}{Property}
\algnewcommand\algorithmicforeach{\textbf{for each}}

\algdef{S}[FOR]{ForEach}[1]{\algorithmicforeach\ #1\ \algorithmicdo}

\input{widebar_hack.tex}

\newcommand{\pnext}{.\text{nx}}
\newcommand{\pprev}{.\text{pr}}

\begin{document}

\title{Gaussian 1-2-1 Networks:\\ Capacity Results for mmWave Communications} 
\author{
\IEEEauthorblockN{Yahya H. Ezzeldin$^\dagger$, Martina Cardone$^{\star}$, Christina Fragouli$^{\dagger}$, Giuseppe Caire$^*$}\\
$^{\dagger}$ UCLA, Los Angeles, CA 90095, USA,
Email: \{yahya.ezzeldin, christina.fragouli\}@ucla.edu\\
$^{\star}$ University of Minnesota, Minneapolis, MN 55404, USA,
Email: cardo089@umn.edu\\
$^*$ Technische  Universit\"{a}t  Berlin,
Berlin, Germany, 
Email: caire@tu-berlin.de
}
\IEEEoverridecommandlockouts
\maketitle

\begin{abstract}
This paper proposes a new model for wireless relay networks referred to as ``1-2-1 network'',  where two nodes can communicate only if they point ``beams'' at each other, while if they do not point beams at each other, no signal can be exchanged or interference can be generated. This model is motivated by millimeter wave communications where, due to the high pathloss, a link between two nodes can exist only if beamforming gain at both sides is established, while in the absence of beamforming gain the signal is received well below the thermal noise floor. The main result in this paper is that the 1-2-1 network capacity can be approximated by routing information along at most $2N+2$ paths, where $N$ is the number of relays connecting a source and a destination through an arbitrary topology. 

\end{abstract}
\section{Introduction}
Millimeter Wave (mmWave) communications are expected to play a vital role in 5G mobile communications, expanding the available spectrum and enabling multi-gigabit services that range from ultra-high definition video, to outdoor mesh networks, to autonomous vehicle platoons and drone communication~\cite{alliance20155g}. Although several works examine channel modeling for mmWave networks~\cite{rappaport2015wideband}, the information theoretic capacity of mmWave relay networks is yet relatively unexplored. 
In this paper, we present capacity results for a class of networks that we term 1-2-1 networks that offer a simple yet informative model for mmWave networks. 

The inherent characteristic of mmWave communications that our model captures is  directivity:   mmWave requires beamforming with narrow beams to compensate for high path loss.
To establish a communication link, both the mmWave transmitter and receiver employ antenna arrays that they electronically steer to direct their beams towards each other - we term this a 1-2-1 link, as both nodes need to focus their beams to face each other for the link to be active.   
Thus, in 1-2-1 networks, instead of broadcasting or interference, we have coordinated steering of transmit and receive beams to activate different links at each time.
An example of a diamond network with $N=4$ relays is shown in Fig.~\ref{fig:example_ntwk}, where two different states for the configuration of the transmit/receive beams are depicted and the resulting activated links are highlighted.

    Our main results are as follows. We consider a source connected to a destination through an arbitrary topology of $N$ relay nodes, and derive a min-cut Linear Program (LP) that outer-bounds the capacity within a constant gap, which only depends on $N$; we then show that its dual is equivalent to a fractional path utilization LP. That is, we show that routing is a capacity achieving strategy (up to a constant gap). Moreover, out of an exponential number 
of paths that potentially connect the source to the destination, we  show we need to utilize at most $2N+2$ to approximately achieve the capacity.  We also prove tighter results for classes of networks. For example,
for the special case of diamond (or one-layer) networks, where the source is connected to the destination through one layer of non-interfering relays as in Fig.~\ref{fig:example_ntwk}, we prove that we can approximately achieve  the network capacity by routing information along {\em at most two paths}, independently of the total number $N$ of relays.
As a result, selecting to operate the best path always achieves half the diamond network capacity.

    \begin{figure}
        \centering
        \includegraphics[width=0.5\textwidth]{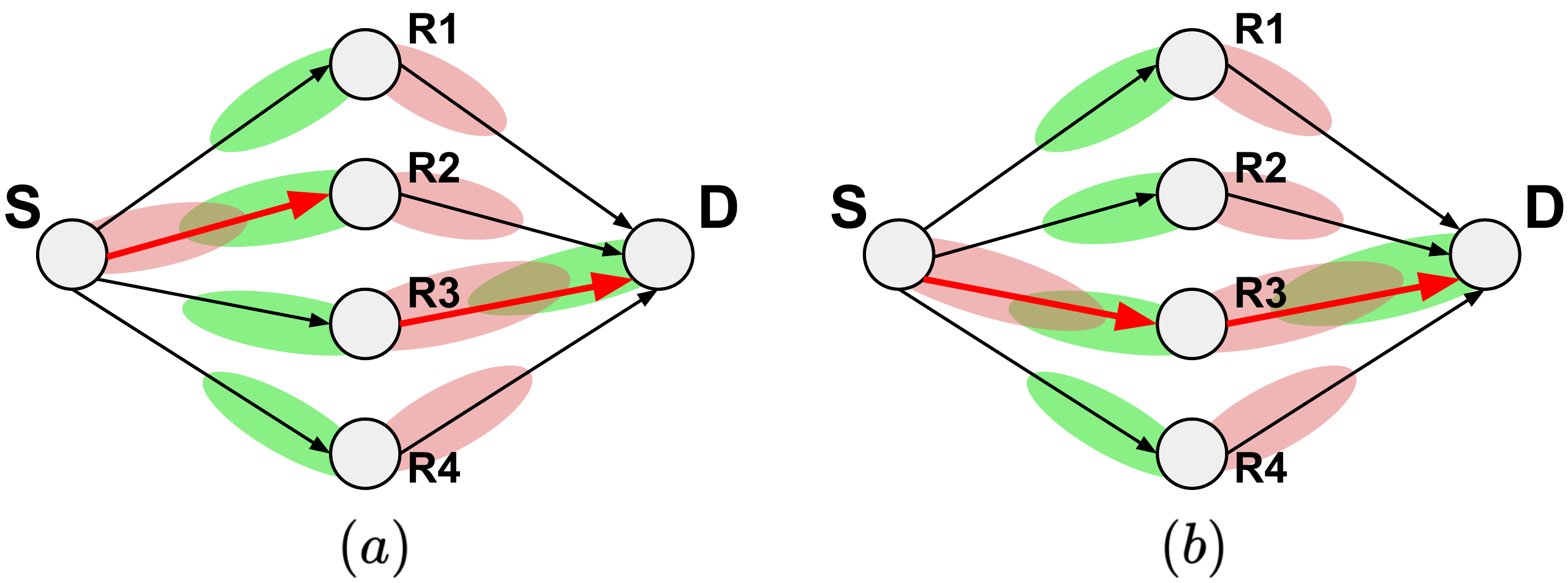}
        \caption{1-2-1 network with $N = 4$ relays and two states.}
        \vspace{-1em}
        \label{fig:example_ntwk}
    \end{figure}

Although we believe that 1-2-1 networks capture the essence of mmWave networks and enable to build useful insights on near-optimal information flow algorithms, we recognize that this model makes a number of simplifying assumptions that include: 1) we assume no interference among communication links  (a reasonable assumption for relays spaced further apart than the beam width), and 2) we do not take into account the overhead of channel knowledge, and of beam-steering.

\smallskip

\noindent{\bf{Related Work.}} Several studies examine channel modeling for mmWave networks~\cite{rappaport2015wideband,akdeniz2014millimeter}. However, to the best of our knowledge, the information theoretic capacity under optimal scheduling has not been analyzed. Recent studies in networking design communication protocols for mmWave mesh networks~\cite{shokri2015millimeter,niu2015survey}. Closer to this work are perhaps works that examine  directional networks in the Gupta\&Kumar framework~\cite{GuptaKumar2000}, however they only look at order arguments for multiple unicast sessions~\cite{yi2003capacity}, and do not consider schedules that arrange for both receiver and transmitter beams to align. 


\smallskip

\noindent{\bf{Paper Organization.}} 
Section~\ref{sec:model} describes the $N$-relay Gaussian 1-2-1 network and derives a constant gap approximation of its capacity;
Section~\ref{sec:FDHandshakeNet} presents our main results;
Section~\ref{sec:ProofMainTh} contains one of the main proofs of this work.

\section{System Model and Capacity Formulation}
\label{sec:model}
With $[n_1:n_2]$ we denote the set of integers from $n_1$ to $n_2 \geq n_1$; $\text{Card}(S)$ is the cardinality of the set $S$; $\emptyset$ is the empty set; $\1_{P}$ is the indicator function;
$0^N$ indicates the all-zero vector of length $N$.

We consider an $N$-relay Gaussian {\em{1-2-1}} network where $N$ relays assist the communication between a source node (node~$0$) and a destination node (node~$N+1$).
In particular, in this 1-2-1 network, at any particular time, a node in the network can only direct (beamform) its transmissions towards at most another node. 
Similarly, a node can only receive transmissions from at most another node (to which its receiving beam points towards).
Thus, each node $i \in [0:N+1]$ in the network is characterized by two states, namely $S_{i,t}$ and $S_{i,r}$ that represent the node towards which node $i$ is beamforming its transmissions and the node towards which node $i$ is pointing its receiving beam, respectively.
In particular, $\forall i \in [0:N+1]$, we have that
\begin{subequations}
\label{eq:state}
\begin{align}
\begin{array}{ll}
S_{i,t} \subseteq [1:N+1]  \backslash \{i \}, & \text{Card}(S_{i,t}) \leq 1,
\\ S_{i,r} \subseteq [0:N]  \backslash \{i \}, & \text{Card}(S_{i,r}) \leq 1,
\end{array}
\end{align}
where $S_{0,r} = S_{N+1,t} = \emptyset$ since the source node always transmits and the destination node always receives.
We consider two modes of operation at the relays, namely {\em{Full-Duplex}} (FD) and {\em{Half-Duplex}} (HD). 
In FD, relay $i \in [1:N]$ can be simultaneously receiving and transmitting, i.e., we can have both $S_{i,t} \neq \emptyset$ and $S_{i,r} \neq \emptyset$.
In HD, relay $i \in [1:N]$ can either receive or transmit, i.e., if $S_{i,t} \neq \emptyset$, then $S_{i,r} = \emptyset$ and vice versa.
In particular, $\forall i \in [1:N]$, we have that
\begin{align}
\text{Card}(S_{i,t}) \!+\! \text{Card}(S_{i,r}) \! \leq \!\!
\left \{
\begin{array}{ll}
\!\!\!2 & \!\text{if relays operate in FD}
\\
\!\!\!1 & \!\text{if relays operate in HD}
\end{array}\!\!.
\right.
\end{align}
\end{subequations}
We can now write the memoryless channel model for this Gaussian 1-2-1 network. We have that $ \forall j \in [1:N+1]$
\begin{align}
    Y_j = Z_j + \sum_{i \in [0:N]\backslash\{j\}} h_{ji} \1_{\{i \in S_{j,r},  \ j \in S_{i,t}\}} X_i,
    \label{eq:model_1}
\end{align}
where: (i) $S_{i,t}$ and $S_{i,r}$ are defined in~\eqref{eq:state};
(ii) $X_i$ (respectively, $Y_i$) denotes the channel input (respectively, output) at node $i$; 
(iii) $h_{ji} \in \mbb{C}$ represents the complex channel coefficient from node $i$ to node $j$; the channel coefficients are assumed to be constant for the whole transmission duration and known by the network;
(iv) the channel inputs are subject to an individual power constraint, i.e., $\E[|X_k|^2] \leq P,\ k \in [0:N]$;
(v) $Z_j,\ j \in [1:N+1]$ indicates the additive white Gaussian noise at the $j$-th node; noises across the network are assumed to be independent and identically distributed as $\mcal{CN}(0,1)$.

By using a similar approach as the one proposed in~\cite{KramerAllerton2004}, the 
channel model in~\eqref{eq:model_1} can be modified to incorporate the state variables in the channel inputs.
In particular, let the vector $\widehat{X}_i = (S_i,\widebar{X}_i)$ be the input to the channel at node $i \in[0:N]$, where: (i) $S_i = (S_{i,t}, S_{i,r})$ with $S_{i,t}$ and $S_{i,r}$ being defined in~\eqref{eq:state} and (ii)
$\widebar{X}_i \in \mbb{C}^{N+1}$, with elements $\widebar{X}_i(k)$ defined as
\begin{align}
\widebar{X}_i(k) = X_i \1_{\{k \in S_{i,t}\}}.
    \label{eq:model_2_inp}
\end{align}
In other words, $\widebar{X}_i$ as a vector is a function of $S_{i,t}$ and the input of the original channel $X_i$. 
When node $i$ is not transmitting, i.e., $S_{i,t} = \emptyset$, then $\widebar{X}_i = 0^{N+1}$.
It is not hard to see that the power constraint on $X_i$ extends to $\widebar{X}_i$ since at most one single index appears in the vector (recall that $\text{Card}(S_{i,t}) \leq 1$).
Using this new channel input $\widehat{X}_i$, we can now equivalently rewrite the channel model in~\eqref{eq:model_1} as
\begin{align}
    Y_j  = 
    \begin{cases}
        h_{j S_{j,r}} \widebar{X}_{S_{j,r}}(j) + Z_j & \text{if}\ \text{Card} (S_{j,r}) = 1\\
        0 & \text{otherwise}
    \end{cases}.
    \label{eq:model_2}
\end{align}
The capacity\footnote{We use standard definitions for codes, achievable rates and capacity.} $\mathsf{C}$ of the network defined in~\eqref{eq:model_2_inp} and~\eqref{eq:model_2} is not known, but can be approximated to within a constant-gap as stated in the following theorem which is proved in Appendix~\ref{app:ConstGap}.
\begin{thm}
The capacity $\mathsf{C}$ of the network defined in~\eqref{eq:model_2_inp} and~\eqref{eq:model_2} can be lower and upper bounded as
\begin{subequations}
\label{eq:constGap}
\begin{align}
&\msf{C}_{\rm cs,iid} \leq \mathsf{C} \leq \msf{C}_{\rm cs,iid} + \mathsf{GAP},
\\
&\msf{C}_{\rm cs,iid} \!=\! \max_{\substack{\lambda_s: \lambda_s \geq 0 \\ \sum_s \lambda_s = 1}} \min_{\Omega \subseteq [1:N] \cup \{0\}} \! \sum_{\substack{(i,j): i \in \Omega,\\ j \in \Omega^c}}\left( \sum_{\substack{ s:\\ j \in s_{i,t},\\ i \in s_{j,r} }} \lambda_s \right) \ell_{j,i}, \label{eq:apprCap}
\\
& \ell_{j,i} = \log\left(1 +P\left| h_{ji}\right|^2\right),
\\
&
\mathsf{GAP} = \mathsf{G}_1 +\mathsf{G}_2 +\mathsf{G}_3  \nonumber
\\& \qquad = (N\!+\!1)\log e \!+\! 2\log(N{+}2) \!+\! N\log(\text{Card}(S_1)), \label{eq:GAP}
\end{align}
where:
(i) $\Omega^c = [0:N+1] \backslash \Omega$; 
(ii) $\lambda_s = \pr(S_{[0:N+1]} = s)$ is the joint distribution of the states, where $s$ enumerates the possible network states $S_{[0:N+1]}$;
(iii) $\text{Card}(S_1)$ is defined as
\begin{align}
\label{eq:RVStateRelays}
\text{Card}(S_1) = \left \{
\begin{array}{ll}
(N+1)^2 &  \text{if relays operate in FD}\\
2N+1 &   \text{if relays operate in HD}
\end{array}.
\right.
\end{align}
\end{subequations}
\end{thm}

The variable $\mathsf{GAP}$ in~\eqref{eq:GAP} only depends on the number of relays $N$ and represents the maximum loss incurred by using independent inputs and deterministic schedules at the nodes.
In particular, $\mathsf{G}_1$ represents the beamforming loss due to the use of independent inputs, while $\mathsf{G}_2$ (respectively, $\mathsf{G}_3$) accounts for the loss  incurred by using a fixed schedule at the source and destination (respectively, at the relays), as we explain next.
Note that from~\eqref{eq:model_2_inp}, the input at the $i$-th node is also characterized by the random state variable $S_{i,t}$
(which indicates to which node -- if any -- node $i$ is transmitting). 
Therefore, information can be conveyed from the source to the destination by randomly switching between 
these states.
However, as first highlighted in~\cite{KramerAllerton2004} in the context of the HD relay channel, this random switch can only improve the capacity by a constant, whose maximum value equals the logarithm of the cardinality of the support of the state random variable.
It therefore follows that the capacity can be approximated to within this constant by using a fixed/deterministic schedule at the nodes.
In particular, for the source and destination the cardinality of the support of their state random variable equals $N+2$ (since the source and the destination can only be transmitting to and receiving from at most one node, respectively). 
Differently, the cardinality of the support of the state variable at the relays depends on the mode of operation (either FD and HD) and is given by~\eqref{eq:RVStateRelays}.

In other words, $\msf{C}_{\rm cs,iid}$ in~\eqref{eq:constGap} -- which can be achieved using QMF as in~\cite{OzgurIT2013} or NNC as in~\cite{NNC} -- is a constant gap away from the capacity $\mathsf{C}$ of the network defined in~\eqref{eq:model_2_inp} and~\eqref{eq:model_2}.
Thus, in the rest of the paper we analyze $\msf{C}_{\rm cs,iid}$, which we refer to as the {\em approximate} capacity for the Gaussian 1-2-1 network.
%
\begin{rem}
\label{rem:RatLinkCap}
{\rm
In the rest of the paper, we assume that the point-to-point link capacities are rational numbers.
In fact, as we prove in what follows, when the capacities $\ell_{j,i} \in \mathbb{R}$, we can always further bound the approximate capacity $\msf{C}^{\mathbb{R}}_{\rm cs,iid} $ as
\begin{align}
\msf{C}^{\mathbb{Q}}_{\rm cs,iid} \leq \msf{C}^{\mathbb{R}}_{\rm cs,iid} \leq \msf{C}^{\mathbb{Q}}_{\rm cs,iid} + \epsilon,
\label{eq:RealRat}
\end{align}
where $\epsilon > 0$ and $\msf{C}^{\mathbb{Q}}_{\rm cs,iid}$ is the approximate capacity of a network with link capacities $ \hat{\ell}_{j,i}$ such that
\begin{align}
\label{eq:RealCap}
\hat{\ell}_{j,i} \leq \ell_{j,i} \leq \hat{\ell}_{j,i} + \frac{\epsilon}{(N+1)^2}, \ \hat{\ell}_{j,i} \in \mathbb{Q}.
\end{align}
Note that such an assignment always exists since the set of rationals $\mathbb{Q}$ is dense in $\mathbb{R}$.
With this, we have
\begin{align*}
\msf{C}^{\mathbb{R}}_{\rm cs,iid} &= \max_{\substack{\lambda: \lambda \geq 0 \\ \sum_s \lambda_s = 1}}  \min_{\Omega \subseteq [1:N] \cup \{0\}}  \sum_{\substack{(i,j): i \in \Omega,\\ j \in \Omega^c}}\left( \sum_{\substack{ s:\\ j \in s_{i,t},\\ i \in s_{j,r} }} \lambda_s \right) \ell_{j,i}\\
&\stackrel{\eqref{eq:RealCap}}{\leq} \max_{\substack{\lambda: \lambda \geq 0 \\ \sum_s \lambda_s = 1}}  \min_{\Omega \subseteq [1:N] \cup \{0\}}  \sum_{\substack{(i,j): i \in \Omega,\\ j \in \Omega^c}}\left( \sum_{\substack{ s:\\ j \in s_{i,t},\\ i \in s_{j,r} }} \lambda_s \right) \left(\hat{\ell}_{j,i} + \frac{\epsilon}{(N+1)^2} \right) \\
& = \max_{\substack{\lambda: \lambda \geq 0 \\ \sum_s \lambda_s = 1}}  \min_{\Omega \subseteq [1:N] \cup \{0\}}
\left \{ \sum_{\substack{(i,j): i \in \Omega,\\ j \in \Omega^c}}\left( \sum_{\substack{ s:\\ j \in s_{i,t},\\ i \in s_{j,r} }} \lambda_s \right) \hat{\ell}_{j,i}
+ 
\sum_{\substack{(i,j): i \in \Omega,\\ j \in \Omega^c}}\left( \sum_{\substack{ s:\\ j \in s_{i,t},\\ i \in s_{j,r} }} \lambda_s \right)\frac{\epsilon}{(N+1)^2} 
\right \} \\
& \stackrel{\sum_{s} \lambda_s \leq 1}{\leq} 
\max_{\substack{\lambda: \lambda \geq 0 \\ \sum_s \lambda_s = 1}}  \min_{\Omega \subseteq [1:N] \cup \{0\}}
\left \{ \sum_{\substack{(i,j): i \in \Omega,\\ j \in \Omega^c}}\left( \sum_{\substack{ s:\\ j \in s_{i,t},\\ i \in s_{j,r} }} \lambda_s \right) \hat{\ell}_{j,i}
+ 
\sum_{\substack{(i,j): i \in \Omega,\\ j \in \Omega^c}}\frac{\epsilon}{(N+1)^2}
\right \} \\
& \leq
\max_{\substack{\lambda: \lambda \geq 0 \\ \sum_s \lambda_s = 1}}  \min_{\Omega \subseteq [1:N] \cup \{0\}}
 \sum_{\substack{(i,j): i \in \Omega,\\ j \in \Omega^c}}\left( \sum_{\substack{ s:\\ j \in s_{i,t},\\ i \in s_{j,r} }} \lambda_s \right) \hat{\ell}_{j,i} + \epsilon = \msf{C}^{\mathbb{Q}}_{\rm cs,iid} + \epsilon .
\end{align*}
Moreover, we have that
\begin{align*}
\msf{C}^{\mathbb{R}}_{\rm cs,iid} &= \max_{\substack{\lambda: \lambda \geq 0 \\ \sum_s \lambda_s = 1}}  \min_{\Omega \subseteq [1:N] \cup \{0\}}  \sum_{\substack{(i,j): i \in \Omega,\\ j \in \Omega^c}}\left( \sum_{\substack{ s:\\ j \in s_{i,t},\\ i \in s_{j,r} }} \lambda_s \right) \ell_{j,i}\\
&\stackrel{\eqref{eq:RealCap}}{\geq}\max_{\substack{\lambda: \lambda \geq 0 \\ \sum_s \lambda_s = 1}}  \min_{\Omega \subseteq [1:N] \cup \{0\}}  \sum_{\substack{(i,j): i \in \Omega,\\ j \in \Omega^c}}\left( \sum_{\substack{ s:\\ j \in s_{i,t},\\ i \in s_{j,r} }} \lambda_s \right) \hat{\ell}_{j,i} = \msf{C}^{\mathbb{Q}}_{\rm cs,iid}.
\end{align*}
This proves~\eqref{eq:RealRat} and hence, in the rest of the paper, we will assume that the point-to-point link capacities are rational numbers.
}
\end{rem}

\section{Main Results}
\label{sec:FDHandshakeNet}
We here present our main results on Gaussian 1-2-1 networks and discuss their implications.
Our first main result
is that for FD networks, $\msf{C}_{\rm cs,iid}$ in~\eqref{eq:apprCap} can be computed  as the sum of fractions 
of the FD capacity of the paths in the network. 
\begin{thm}
\label{thm:MainTh}
For any $N$-relay Gaussian FD 1-2-1 network, we have that
\begin{align}
\label{eq:CapPaths}
\begin{array}{llll}
{\rm P1:} &\msf{C}_{\rm cs,iid}  = {\rm max}  \displaystyle\sum_{p \in \mcal{P}} x_p \mathsf{C}_p   & & \\
&      ({\rm P1}a) \ x_p \geq 0 & \forall p \!\in\! \mcal{P}, &  \\
&    ({\rm P1}b) \ \displaystyle\sum_{p \in \mcal{P}_i}  x_p f^p_{p\pnext(i),i} \!\leq\! 1 & \forall i \! \in \! [0\!:\!N], & \\
&   ({\rm P1}c) \ \displaystyle\sum_{p \in \mcal{P}_i} x_p f^p_{i,p\pprev(i)} \!\leq\! 1 & \forall i \! \in \! [1\!:\!N\!+\!1],  &
\end{array}
\end{align}
where: 
(i) $\mcal{P}$ is the collection of all paths from the source to the destination;
(ii) $\mcal{P}_i \subseteq \mathcal{P}$ is the collection of paths that pass through node $i \in [0:N+1]$ (clearly, $\mcal{P}_0=\mcal{P}_{N+1}=\mcal{P}$ since all paths pass through the source and the destination);
(iii) $\mathsf{C}_p$ is the FD capacity of the path $p \in \mcal{P}$, i.e., $\mathsf{C}_p = \min_{(i,j)\in p} \ell_{j,i}$;
(iv) $p\pnext(i)$  (respectively, $p\pprev(i)$) with $i \in [0:N+1]$ is the node following (respectively, preceding) node $i \in [0:N+1]$ in path $p \in \mcal{P}$ (clearly, $p\pprev(0) = p\pnext(N+1) = \emptyset$);
(v) $f^p_{j,i}$ is the optimal activation time for the link of capacity $\ell_{j, i}$ when the path $p \in \mcal{P}$, such that $(i,j)\in p$, is operated, i.e.,
\begin{align}
\label{eq:actTime}
f^p_{j,i} = \frac{\mathsf{C}_{p}}{\ell_{j, i}}.
\end{align}
\end{thm}
\begin{proof}
The proof of Theorem~\ref{thm:MainTh} is delegated to Section~\ref{sec:ProofMainTh}.
\end{proof}
In the LP in Theorem~\ref{thm:MainTh},
the variable $x_p$ represents the fraction\footnote{Note that $x_p$ in P1 implicitly satisfies that $x_p \leq 1,\ \forall p \in \mcal{P}$. This is due to the fact that for any path $p \in \mcal{P}$, the definition of $f_{j,i}^p$ in \eqref{eq:actTime} implies that at least one constraint in (P1$b$) and (P1$c$) has $f_{j,i}^p = 1$.}
of time the path {$p \in \mcal{P}$} is utilized in the network.
Moreover, each
of the constraints in $({\rm P1}b)$ (respectively, $({\rm P1}c)$) ensures that a node $i \in [0:N+1]$ - even though it can appear in multiple paths in the network - does not transmit (respectively, receive) for more than 100\% of the time.

Lemma~\ref{lem:paths_and_caps} follows  from P1 in Theorem~\ref{thm:MainTh}, and states that, although  the number of paths $P$ in general is exponential in the number of relays $N$, we need to use at most a linear number of paths; this can also be translated to a guarantee on the rate that can be achieved when only the best path is operated.
\begin{lem}\label{lem:paths_and_caps}
For any $N$-relay Gaussian FD 1-2-1 relay network, we have the following guarantees:
\begin{enumerate}[(L1)]
    \item For {a network with arbitrary topology,} the approximate capacity $\msf{C}_{\rm cs,iid}$ can {always} be achieved by activating at most $2N+2$ paths in the network.
    \item For {a network with arbitrary topology,} the best path has an FD capacity $\msf{C}_1$ such that $\msf{C}_1 \geq \frac{1}{2N+2}\msf{C}_{\rm cs,iid}$. 
    \item For a 2-layer relay network with $M = N/2$ relays per layer, the approximate capacity $\msf{C}_{\rm cs,iid}$ can be achieved by activating at most $2M+1$ paths in the network.
    \item For a 2-layer relay network with $M = N/2$ relays per layer, the best path has an FD capacity $\msf{C}_1$ such that $\msf{C}_1 \geq \frac{1}{2M+1}\msf{C}_{\rm cs,iid}$.
    .
\end{enumerate}
\end{lem}

\begin{proof}
~

{\bf Proof of L1:} The LP P1 in~\eqref{eq:CapPaths} is bounded and hence there always exists an optimal corner point. In particular, at any corner point in P1, we have at least $P = \text{Card}(\mcal{P})$ constraints satisfied with equality among $({\rm P1}a)$, $({\rm P1}b)$ and $({\rm P1}c)$. Therefore, we have at least $P-2N-2$ in $({\rm P1}a)$ satisfied with equality (since $({\rm P1}b)$ and $({\rm P1}c)$ combined represent $2N+2$ constraints). Thus, at least $P-2N-2$ paths are not operated, which proves the statement in L1.

{\bf Proof of L3:} Similar to the proof above for L1, a corner point in the LP P1 has at most $2N+2$ constraints among $({\rm P}1b)$ and $({\rm P}1c)$ satisfied with equality. To prove L3, we need to show that in the case of a 2-layered network and we have $2N+2$ equality satisfying constraints, then at least one of the equations is redundant.Note that for a 2-layer relay network, any path $p$ in the network has four nodes and is written as $0 - p(1) - p(2) - N+1$ where $p(1)$ and $p(2)$ represent the node in the path from layer 1 and layer 2, respectively. We assume that the relays in the first layer are indexed with $[1:N/2]$ and the second layer relays are indexed with $[N/2+:N]$. Thus, for any path, $p(1) \in [1:N/2]$ and $p(2) \in [N/2+1:N]$. Now assume that all constraints $({\rm P}1b)$ and $({\rm P}1c)$ satisfied with equality, then by adding all $({\rm P}1b)$ constraints for $i \in[1:N/2]$ and subtracting from all constraints from $({\rm P}1c)$ for $j \in [N/2+1:N]$, we get
\begin{align*}
{\rm LHS}:	\sum_{i =1}^{N/2} ({\rm P}1b)_i - \sum_{j =N/2+1}^{N}({\rm P}1c)_j &=\sum_{i =1}^{N/2} \sum_{p \in \mcal{P}_i} x_p f^p_{p\pnext(i),i} - \sum_{j =N/2+1}^{N}\sum_{p \in \mcal{P}_j} x_p f^p_{j,p\pprev(j)}\\
&= \sum_{p \in \mcal{P}} x_p f^p_{p(2),p(1)} - \sum_{p \in \mcal{P}} x_p f^p_{p(2),p(1)} = 0\\
{\rm RHS}:\sum_{i =1}^{N/2} ({\rm P}1b)_i - \sum_{j =N/2+1}^{N}({\rm P}1c)_j &= \sum_{i =1}^{N/2} 1 - \sum_{j =N/2+1}^{N} 1 = 0, &
\end{align*}
Thus, when all constraints $({\rm P}1b)$ and $({\rm P}1c)$ are satisfied with equality, at least one of them redundant, which proves the statement L3.

{\bf Proof of L2 and L4:} The proof of L2 follows directly from L1 by considering only the $2N+2$ paths needed to achieve $\msf{C}_{\rm cs,iid}$ and picking the path that has the largest FD capacity among them.
In particular, the guarantee in L2 is true for the selected path due to the fact that for any feasible point in P1, $x_p \leq 1,\ \forall p \in \mcal{P}$. The proof of L4 from L3 follows the same argument used to prove L2 from L1.
\end{proof}

The result L3 in Lemma~\ref{lem:paths_and_caps} suggests that for 1-2-1 networks with particular structures, we can further reduce the number of active paths needed to achieve the approximate capacity.
In particular, we explore this observation in the context of 1-2-1 Gaussian diamond networks operating in FD and HD, through the following two lemmas proved in Appendix~\ref{app:diamond_FD}, Appendix~\ref{app:diamond_HD} and Appendix~\ref{app:Lemma5L3}.

\begin{lem}\label{lem:diamond_1-2-1_ntwk}
\label{thm:MainThDiam}{}
For the $N$-relay Gaussian diamond 1-2-1 network, we can calculate the approximate capacity $\msf{C}_{\rm cs,iid}$ as
\begin{align}
\label{eq:CapPaths_diamond}
\begin{array}{llll}
{\rm P1^d:} &\msf{C}_{\rm cs,iid}  = {\rm max}  \sum_{p \in [1:N]} x_p \mathsf{C}_p   & & \\
&      ({\rm P1}a)^{\rm d} \ {0 \leq x_p \leq 1} & \forall p \in [1{:}N], &  \\
&    ({\rm P1}b)^{\rm d} \ \sum_{p \in [1:N]}x_p \frac{\msf{C}_p}{\ell_{p,0}} \!\leq\! 1, & & \\
&   ({\rm P1}c)^{\rm d} \ \sum_{p \in [1:N]} x_p \frac{\msf{C}_p}{\ell_{N\!+\!1,p}} \!\leq\! 1, &  &
\end{array}
\end{align}
where: 
(i) $\mcal{P}$ is the collection of all paths from the source to the destination;
(ii) $\mathsf{C}_p$ is the capacity of the path $0 \to p \to N+1$ and its value depends on whether the network is operating in FD or HD, namely
\begin{align*}
\mathsf{C}_p = \left \{
\begin{array}{ll}
\min\{ \ell_{p,0}, \ell_{N+1,p}\} &  \text{if relays operate in FD}\\
 \frac{\ell_{p,0}\ \ell_{N+1,p}}{\ell_{p,0} + \ell_{N+1,p}} &   \text{if relays operate in HD}
\end{array}.
\right.
\end{align*}
\end{lem}

\begin{lem}\label{lem:simplification_diamond}
For an $N$-relay Gaussian FD diamond relay network, we have the following guarantees:
\begin{enumerate}[(L1)]
    \item If the network is operating in FD, then the approximate capacity $\msf{C}_{\rm cs,iid}$ can always be achieved by activating at most $2$ relays in the network, independently of $N$.
    \item If the network is operating in HD, then the approximate capacity $\msf{C}_{\rm cs,iid}$ can always be achieved by activating at most $3$ relays in the network,  independently of $N$.
    \item In both FD and HD networks, the best path has a capacity $\msf{C}_1$ such that $\msf{C}_1 \geq \frac{1}{2}\msf{C}_{\rm cs,iid}$; furthermore, this guarantee is tight for both the FD and HD cases, i.e.,
there exists a class of Gaussian diamond 1-2-1 networks such that $\msf{C}_1 \leq \frac{1}{2}\msf{C}_{\rm cs,iid}$ both for the FD and HD cases.
\end{enumerate}
\end{lem}
The results in L1 and L2 in Lemma~\ref{lem:simplification_diamond} are surprising as they state that, independently of the total number of relays in the network, there always exists a subnetwork of $2$ (in FD) and $3$ (in HD) relays that achieves the full network approximate capacity. Moreover, the guarantee provided by L3 is tight. To see this, consider
$N=2$ and $\ell_{1,0}=\ell_{3,2} = 1$ and $\ell_{3,1}=\ell_{2,0}=X \rightarrow \infty$.
For this network, we have that the approximate capacity is $\msf{C}_{\rm cs,iid}=\ell_{1,0}+\ell_{3,2} = 2$, while the capacity of each path (both in FD and HD) is $\msf{C}_{1}=\min \left \{ 1,X\right \}= 1$, hence $\msf{C}_{1}/\msf{C}_{\rm cs,iid}=1/2$.

\section{Proof of Theorem~\ref{thm:MainTh}}
\label{sec:ProofMainTh}
We here prove Theorem~\ref{thm:MainTh}. 
We note that for a fixed $\lambda_s$, the inner minimization in~\eqref{eq:apprCap} is the standard min-cut problem over a graph with link capacities given by 
\begin{align}
\label{eq:ellMaxFlow}
\ell_{j,i}^{(s)} = \left( \sum_{\substack{ s:\\ j \in s_{i,t},\ i \in s_{j,r} }} \lambda_s \right) \ell_{j,i}.
\end{align}
Since the min-cut problem is the dual for the standard max-flow problem, then we can replace the inner minimization in~\eqref{eq:apprCap} with the max-flow problem over the graph with link capacities defined in~\eqref{eq:ellMaxFlow} to give that
\begin{align}
\label{eq:approx_cap_flow_2}
\text{P2-flow}&{\rm\ :} \ \msf{C}_{\rm cs,iid} = \max_{\substack{\lambda_s: \lambda_s \geq 0 \\ \sum_s \lambda_s = 1}} \max \sum_{j =1}^{N+1} F_{j,0}& \nonumber\\
& 0 \leq F_{j,i} \leq \ell_{j,i}^{(s)}& (i,j) \in [0:N]\times[1:N{+}1], \\
& \sum_{j\in[1:N{+}1]\backslash\{i\}} F_{j,i} = \sum_{k\in[0:N]\backslash\{i\}} F_{i,k} & i \in [1:N],\nonumber 
\end{align}
where $F_{j,i}$ represents the flow from node $i$ to node $j$.

The max-flow problem can be equivalently written as an LP with path flows instead of link flows, thus,~\eqref{eq:approx_cap_flow_2} can be written as P2 described next by using the path flows representation of the max-flow problem. A variable $F_p$ is used
for the flow through the path $p \in \mcal{P}$. 
\begin{align*}
\begin{array}{llll}
{\rm P2}&{\rm :} \ \msf{C}_{\rm cs,iid} = \max \displaystyle\sum_{p \in \mcal{P}} F_p & & \\
& ({\rm P2}a) \ F_p \geq 0 & \forall p \in \mcal{P}, & \\
& ({\rm P2}b) \! \! \! \! \displaystyle \sum_{\substack{p \in \mcal{P},\\(i,j) \in p,\\ j = p\pnext(i)}} \! \! \! \!\!\! F_p {\leq} \ell_{j,i}^{(s)} {=}  {\rm eq.}\eqref{eq:ellMaxFlow} & \forall (j,i)  {\in}  [1{:}N{+}1] \! \times \![0{:}N], & \\
& ({\rm P2}c) \ \sum_{s} \lambda_s \leq 1, & & \\
& ({\rm P2}d) \ \lambda_s \geq 0  & \forall s, &
\end{array}
\end{align*}
The constraints $({\rm P2}b)$ ensure that, for any link from node $i$ to node $j$, the sum of the flows through the paths that use this link does not exceed the link modified capacity $\ell_{j,i}^{(s)}$  in~\eqref{eq:ellMaxFlow}. The constraint $({\rm P2}c)$ is the same constraint on $\lambda_s$ as in~\eqref{eq:apprCap}.

To prove Theorem~\ref{thm:MainTh}, we first show that P2 above is equivalent to the following LP P3, and then prove that P3 is equivalent to P1 in~\eqref{eq:CapPaths}, which completes the proof.
\begin{align*}
\begin{array}{llll}
& {\rm P3:} \ \msf{C}_{\rm cs,iid} = \max \sum_{p \in \mcal{P}} F_p & & \\
& ({\rm P3}a)  \ F_p \geq 0 & \forall p \in \mcal{P}, & \\
& ({\rm P3}b)  \! \! \! \! \displaystyle \sum_{\substack{p \in \mcal{P},\\(i,j) \in p,\\ j = p\pnext(i)}} \! \! \! \! F_p \leq \lambda_{\ell_{j,i}} \ell_{j,i} & \forall (i,j) {\in} [0\!:\!N] \!\times\! [1\!:\!N{+}1], & \\
& ({\rm P3}c)  \! \! \! \! \displaystyle \sum_{\substack{j \in [1:N{+}1]\backslash\{i\}}} \! \! \! \!  \lambda_{\ell_{j,i}} \leq 1 & \forall i \in [0:N], & \\
& ({\rm P3}d)  \! \! \! \! \displaystyle \sum_{\substack{i \in [0:N]\backslash\{j\}}} \! \! \! \!  \lambda_{\ell_{j,i}} \leq 1 & \forall j \in [1:N+1], & \\
& ({\rm P3}e)  \ \lambda_{\ell_{j,i}} \geq 0  &\forall (i,j)  {\in} [0\!:\!N] \!\times\! [1\!:\!N{+}1], &
\end{array}
\end{align*}
where $\lambda_{\ell_{j,i}},(i,j) \in [0:N] \times [1:N+1]$ represents the fraction of time the link of capacity $\ell_{j,i}$ is active.
The constraints $({\rm P3}b)$ are similar to $({\rm P2}b)$ except that the capacity of a link is now
modified 
through a multiplication by $\lambda_{\ell_{j,i}}$.
The constraints $({\rm P3}c)$ ensure that, for any transmitting node $i$, the sum of the activation times of its outgoing links is less than 100\% of the time.
Similarly, $({\rm P3}d)$ ensure the same logic for the incoming edges to a receiving node.

We delegate the proof of the equivalence between P3 and P1 to Appendix~\ref{app:P3EqualP1}, and here  prove the equivalence between P2 and P3, which is more involved.  To do so, we first show that a feasible point in P2 gives a feasible point in P3 with the same objective value. We define the following transformation
\[
    \lambda_{\ell_{j,i}} =  \sum_{\substack{ s:\\ j \in s_{i,t},\ i \in s_{j,r} }} \lambda_s,
\]
while the variable $F_p$ in P2 is the same as the variable $F_p$ in P3. 
By substituting these transformations in the constraints $({\rm P2}a)$-$({\rm P2}d)$, it is not difficult to see that the constructed point (through the transformation) is feasible in P3 and has the same objective value in P3 as the objective value in P2. Thus, the solution of P3 is at least as large as the one of P2.

We now prove the  direction from P3 to P2, by showing that every corner point in P3  can be transformed into a feasible point in P2 with the same objective value. To do this, we introduce a visualization for P2 and P3 in terms of bipartite graphs.  We divide each node $i \in [0:N+1]$ in the network into two nodes ($i_T$ and $i_R$) representing the transmitting and receiving functions of the node;
note that $0_R = (N+1)_T = \emptyset$ since the source (node $0$) is always transmitting and the destination (node $N+1$) is always receiving.
This gives us the bipartite graph $\mcal{G}_B = (\mcal{T},\mcal{R},\mcal{E})$, where the vertices $\mcal{T}$ (respectively, $\mcal{R}$) are the  transmitting modules of our nodes (respectively, $\mcal{R}$ collects our receiving modules), and we have an edge  $(i_T,j_R) \in \mcal{E}$ for each link  in the network. 
It is easy to see that a valid state in P2 represents a matching in the bipartite graph $\mcal{G}_B$.
Furthermore, we can write the program P3 in terms of $\mcal{G}_B$ by simply renaming all $\lambda_{\ell_{j,i}}$ with $\lambda_{j_R,i_T}$, i.e., the activation time of the edge $(i_T,j_R)$ in $\mcal{G}_B$. 

Starting with a corner point in P3, we can follow the procedure described below to construct a feasible point in P2; note that since all coefficients (i.e., $\ell_{j_R,i_T}$) in  P3 are rational (see Remark~\ref{rem:RatLinkCap}), the  corner points $\lambda^\star_{\ell_{j_R,i_T}} \in \mbb{Q},\ \forall (i_T,j_R)$ are rational.  The main intuition is to use the fractions $\lambda^\star_{\ell_{j_R,i_T}}$ and our bipartite graph $\mcal{G}_B$ to construct a bipartite multigraph with edges of unit capacity such that for higher values of $\lambda^\star_{\ell_{j_R,i_T}}$, we will have more parallel edges from the $i_T$-th node to the $j_R$-th node. 
In particular, our procedure consists of the three following main steps.\\
{\bf Step 1.}
We multiply all $\lambda^\star_{\ell_{j_R,i_T}}$ by the Least Common Multiple (LCM) $M$ of their denominators (for simplicity, if $\lambda^\star_{\ell_{j_R,i_T}}=0$, then the denominator is set to be one). 
We then calculate $n_{j,i} = M \lambda^\star_{\ell_{j_R,i_T}},\ \forall i_T,j_R \in [0:N+1]$, and construct the bipartite multigraph $\mcal{G}^\bullet_B$ from $\mcal{G}_B$ that has $n_{j,i}$ parallel edges from node $i_T$ to node $j_R$.\\
{\bf Step 2.} We edge color the multigraph $\mcal{G}^\bullet_B$. 
If any of the parallel edges from node $i_T$ to node $j_R$ are colored with color $c_i$, we say that $c_i$ activates the link $(i_T,j_R)$ in $\mcal{G}_B$ (or equivalently the link $(i,j)$ in the network).
Note that since $\mcal{G}^\bullet_B$ is a bipartite graph, then there is an optimal coloring for the graph that uses $\Delta$ colors, where $\Delta$ is the maximum degree of the nodes.
Furthermore, since no two adjacent edges in $\mcal{G}^\bullet_B$ can have the same color, every individual color represents a matching in the bipartite graph $\mcal{G}^\bullet_B$ (and by extension the graph $\mcal{G}_B$). 
Each of these matchings represents a state in the network; since there are $\Delta$ colors in total, then we assign to state $s_i,\ i \in [1:\Delta]$ (represented by color $c_i$) an activation time of $w_i = 1/\Delta$. \\
{\bf Step 3.} From Step~2, we know that each color in the network represents a matching in $\mcal{G}_B$, and hence a state. However, some colors can correspond to the same matching in $\mcal{G}_B$, i.e., two or more colors activate exactly the same set of links. We combine these colors together into one state by multiplying $1/\Delta$ by the number of times this state is repeated. The remaining unique states represent the feasible states in P2 that were obtained from our optimal $\lambda^\star$ from P3. 

We now need to prove that the states that we obtain from the previous steps, in addition to the flows through the paths that we have from P3, indeed give us a feasible point in P2. 
Without loss of generality, we will prove that the states generated in Step 2 give a feasible point since combining similar states (Step~3) does not change the total sum of activation times and does not change the amount of time a link is active (which is what we look for in the constraint~$({\rm P2}b)$.

To show the feasibility of the constructed schedule, we first derive a consequence of the fact that $\lambda^\star_{\ell_{j_R,i_T}}$ is feasible in P3.
	In particular, we can show the following inequality between the maximum degree $\Delta$ and the constant LCM $M$
\begin{align}\label{eq:Delta_2_M}
\Delta &=\!\max\left[\max_{i \in [0:N]}\left( \sum_{j \in [1:N+1]} \! n_{j,i}\! \right), \max_{j \in [1:N+1]}\left( \sum_{i \in [0:N]} \! n_{j,i} \right)  \right] \nonumber\\
&= \!M \max\left[ \max_{i_T \in [0:N]}\left( \sum_{j_R \in [1:N+1]} \lambda^\star_{\ell_{j_R,i_T}} \right), \right. \nonumber
\\ & \left. \quad \qquad  \max_{j_R \in [1:N+1] }\left( \sum_{i_T \in [0:N]} \lambda^\star_{\ell_{j_R,i_T}} \right)  \right] \stackrel{({\rm P3}c,d)}\leq M,
\end{align}
where the inequality follows from the constraints $({\rm P3}c)$ and $({\rm P3}d)$ in P3.
Using~\eqref{eq:Delta_2_M}, we can now show that the constructed schedule (that uses coloring arguments) is feasible in P2 (note that since $F_p$ is unchanged then the constraint in~$({\rm P2}a)$ is already satisfied due to the constraint in~{$({\rm P3}a)$}). Moreover, we have
\begin{align*}
&({\rm P3}b): \   \forall (j,i), \ \sum_{\substack{p \in \mcal{P},\ (i,j) \in p,\\ j = p.\pnext(i)}} F_p \leq \lambda_{\ell_{j,i}} \ell_{j,i} 
= \frac{n_{j,i}}{M} \ell_{j,i} 
\\& \hspace{0.3cm} = \left( \displaystyle\sum_{s: (i,j) \in s}\lambda_s\right)\frac{\Delta}{M}\ell_{j,i} \stackrel{\eqref{eq:Delta_2_M}}\leq \left( \displaystyle\sum_{s:(i,j) \in s}\lambda_s\right)\ell_{j,i} \implies ({\rm P2}b)\\
&\text{By construction:} \ \sum_s \lambda_s = \sum_{s} \frac{1}{\Delta} = \Delta \frac{1}{\Delta} = 1 \implies ({\rm P2}c)\\
&\text{By construction:} \ \lambda_s = \frac{1}{\Delta} \geq 0 \implies ({\rm P2}d).
\end{align*}
From the discussion above, it follows that we can map a rational  point $\lambda^\star_{\ell_{j_R,i_T}}$ in P3 to a feasible point in P2 that has the same objective function value. 
Note that the variables $F_p$ in P3 and P2 are unchanged and therefore, the objective values will remain the same.
Thus, the solution of P2 is at least as large as that of P3, concluding the proof that P2 and P3 are equivalent.

{\begin{rem}
\label{rem:Schedule}
{\rm
The procedure described earlier gives a non-polynomial approach to construct an optimal schedule for the approximate capacity of the Gaussian FD 1-2-1 network. 
In Appendix~\ref{app:poly_algorithm}, we provide an algorithm that computes the optimal schedule as well as the approximate capacity in polynomial time (in the number of nodes).
}
\end{rem}
}

\appendices

\section{Constant Gap Capacity Approximation for the Gaussian 1-2-1 Network}
\label{app:ConstGap}
The memoryless model of the channel allows to upper bound the channel capacity $\mathsf{C}$ using the cut-set upper bound $\mathsf{C}_{\rm cs}$ as
\begin{align}
    \mathsf{C}_{\rm cs} &= \max_{\pr_{\{\widebar{X}_i,S_i\}}(\cdot)} \ \min_{\Omega \subseteq [1:N] \cup \{0\}} I(\widehat{X}_\Omega; Y_{\Omega^c} | \widehat{X}_{\Omega^c})\nonumber \\
    &= \max_{\pr_{\{\widebar{X}_i,S_i\}}(\cdot)} \ \min_{\Omega \subseteq [1:N] \cup \{0\}} I(S_\Omega, \widebar{X}_\Omega; Y_{\Omega^c} | S_{\Omega^c}, \widebar{X}_{\Omega^c})\nonumber \\
    &= \max_{\pr_{\{\widebar{X}_i,S_i\}}(\cdot)} \ \min_{\Omega \subseteq [1:N] \cup \{0\}} I(\widebar{X}_\Omega; Y_{\Omega^c} | S_{\Omega}, S_{\Omega^c}, \widebar{X}_{\Omega^c}) + I(S_\Omega; Y_{\Omega^c} | S_{\Omega^c}, \widebar{X}_{\Omega^c})\nonumber \\
    &\leq \max_{\pr_{\{\widebar{X}_i,S_i\}}(\cdot)} \ \min_{\Omega \subseteq [1:N] \cup \{0\}} I(\widebar{X}_\Omega; Y_{\Omega^c} | S_{[0:N+1]}, \widebar{X}_{\Omega^c}) + H(S_\Omega) \nonumber \\
    &\stackrel{(a)}\leq \max_{\pr_{\{\widebar{X}_i,S_i\}}(\cdot)} \ \min_{\Omega \subseteq [1:N] \cup \{0\}} I(\widebar{X}_\Omega; Y_{\Omega^c} | S_{[0:N+1]}, \widebar{X}_{\Omega^c}) + 2\log(N+2) + N\log( \text{Card} (S_1)  ) \nonumber\\
    &{\stackrel{(b)}=} \max_{\pr_{\{S_i\}}(\cdot)} \max_{\pr_{\{\widebar{X}_i\}|\{S_i\}}(\cdot)}\ \min_{\Omega \subseteq [1:N] \cup \{0\}} \sum_s \lambda_s\ I(\widebar{X}_\Omega; Y_{\Omega^c} | S_{[0:N+1]} {=} s, \widebar{X}_{\Omega^c}) {+} 2\log(N+2) {+} N\log( \text{Card} (S_1)  ) \nonumber \\
    &{\stackrel{(c)}\leq} \max_{\pr_{\{S_i\}}(\cdot)} \ \min_{\Omega \subseteq [1:N] \cup \{0\}} \max_{\pr_{\{\widebar{X}_i\}|\{S_i\}}(\cdot)}\sum_s \lambda_s\ I(\widebar{X}_\Omega; Y_{\Omega^c} | S_{[0:N+1]} {=} s, \widebar{X}_{\Omega^c}) {+} 2\log(N+2) {+} N\log( \text{Card} (S_1)  ),
    \label{eq:cutset_1}
\end{align}
where: 
(i) $\Omega^c = [0:N+1] \backslash \Omega$; 
(ii) $\pr_{\{\widebar{X}_i,S_i\}}(\cdot)$ is the probability distribution of the channel input $\{(S_i,\widebar{X}_i)\}_{i=0}^{N+1}$; 
(iii) $S_{\Omega} = \left \{ S_i | i \in \Omega \right \}$;
(iv) the inequality in $(a)$ is due to the fact that the state variable at the source and destination can take $N+2$ values (since the source can only be transmitting to at most one node and the destination can only be receiving from at most one node), while at each relay the state variable can take $\text{Card}(S_1)$ values, where $\text{Card}(S_1)$ depends on the mode of operation at the relays, namely
\begin{align*}
\text{Card}(S_1) = \left \{
\begin{array}{ll}
(N+1)^2 &  \text{if relays operate in FD}\\
2N+1 &   \text{if relays operate in HD}
\end{array};
\right.
\end{align*}
(v) in the equality in $(b)$ we use $s$ to enumerate the possible network states $S_{[0:N+1]}$ and we denote with $\lambda_s = \pr(S_{[0:N+1]} = s)$ the joint distribution of the states; (vi) the inequality in $(c)$ follows from the max-min inequality.

For a network state $s$, we define the channel matrix $\widehat{H}_s$, where the element $[\widehat{H}_s]_{i,j}$ is defined as
\begin{align}
    [\widehat{H}_s]_{i,j} = 
    \begin{cases}
        h_{ij} & \text{if}\ i \in s_{j,t}\ \text{and}\ j \in s_{i,r} \\
        0 & \text{otherwise},
    \end{cases}
    \label{eq:H_s}
\end{align}
where $h_{ij}$ is the channel coefficient of the link from node $j$ to node $i$.
It is not difficult to see that every row (and column) of $\widehat{H}_s$ has at most one non-zero element and thus there exists a permutation matrix $\Pi$ such that $\Pi\widehat{H}_s$ is a diagonal matrix.
Also, let $s^+ = \left \{ i | s_{i,t} \neq \emptyset, \forall i \in [0:N] \right \}$ and $s^- = \left \{ i | s_{i,r} \neq \emptyset, \forall i \in [1:N+1] \right \}$.

With this, we can further simplify the mutual information expression in~\eqref{eq:cutset_1} as follows
\begin{align}
 & \max_{\pr_{\{\widebar{X}_i\}|\{S_i\}}(\cdot)} \sum_s \lambda_s\ I (\widebar{X}_\Omega; Y_{\Omega^c} | S_{[0:N+1]} {=} s, \widebar{X}_{\Omega^c})\nonumber\\
 &\stackrel{(a)}=\max_{\pr_{\{\widebar{X}_i\}|\{S_i\}}(\cdot)} \sum_s \lambda_s\ I(\widebar{X}_{s^+, \Omega}; Y_{s^-,\Omega^c} | S_{[0:N+1]} {=} s, \widebar{X}_{\Omega^c}) \nonumber \\
 &\stackrel{(b)}= \sum_s \lambda_s \log\det\left(I + \widehat{H}_{s,\Omega}\ K_{s,\Omega}\ \widehat{H}^H_{s,\Omega}\right)\nonumber\\
 &=\sum_s \lambda_s \log\det\left(I + \widehat{H}^H_{s,\Omega}\widehat{H}_{s,\Omega}\ K_{s,\Omega}\ \right),
    \label{eq:mutual_info_gaussian}
\end{align}
where: (i) we define $\widebar{X}_{s^+,\Omega}$ as $\widebar{X}_{s^+,\Omega} = \left\{\left.\widebar{X}_i(s_{i,t}) \right| i \in \Omega \cap s^+ \right\}$ and $\widebar{Y}_{s^-,\Omega^c}$ as $\widebar{Y}_{s^-,\Omega^c} = \left\{\left.\widebar{Y}_i \right| i \in \Omega^c \cap s^- \right\}$;
        (ii) the equality in $(a)$ follows since, given the state $s$, all variables $\widebar{X}_i(j)$, with $j \neq s_{i,t}$, as well as all $Y_i$ with $s_{i,r} =\emptyset$ are deterministic;
        (iii) the equality in $(b)$ follows due to the maximization of the mutual information by the Gaussian distribution;
        (iv) $\widehat{H}_{s,\Omega}$ is a submatrix of $\widehat{H}_s$ (defined in~\eqref{eq:H_s}) and is defined as $\widehat{H}_{s,\Omega} = [\widehat{H}_s]_{\Omega^c,\Omega}$ and $K_{s,\Omega}$ is the submatrix of the covariance matrix of the random vector $\left[\bar{X}_0(s_{0,t})\ \bar{X}_1(s_{1,t})\dots \bar{X}_N(s_{N,t}) \right]^T$, where the rows and columns are indexed by $\Omega$. 

        We now further upper bound the Right-Hand Side (RHS) of~\eqref{eq:mutual_info_gaussian} using~\cite[Lemma 1]{NNC}, for any $\gamma \geq e-1$ as follows
        \begin{align}
            \log\det\left(I + \widehat{H}^H_{s,\Omega}\widehat{H}_{s,\Omega}\ K_{s,\Omega}\right) &\leq \log\det\left(I + \gamma^{-1}P \widehat{H}^H_{s,\Omega}\widehat{H}_{s,\Omega}\right) + |\Omega|\log\alpha(\Omega,s,\gamma) \nonumber \\
            &\stackrel{(a)}\leq \log\det\left(I +P\widehat{H}_{s,\Omega}\widehat{H}^H_{s,\Omega}\right) + |\Omega|\log\alpha(\Omega,s,\gamma),
            \label{eq:lemma_1}
        \end{align}
where the inequality in $(a)$ follows since $\gamma > 1$ and by applying Sylvester's determinant identity and  $\alpha(\Omega,s,\gamma)$ is defined based on~\cite[Lemma 1]{NNC} as
        \begin{align}
            \alpha(\Omega,s,\gamma)=
            \begin{cases}
                e^{\gamma/e} & \text{if}\ \gamma \leq e \frac{{\rm rank}(H_{s,\Omega})}{{\rm trace}(K_{s,\Omega}/P)} = e \frac{ {\rm rank}(H_{s,\Omega})}{|s^+ \cap \Omega|}\\
                \left(\gamma \frac{|s^+ \cap\Omega|}{ {\rm rank}(H_{s,\Omega})}\right)^\frac{ {\rm rank}(H_{s,\Omega})}{|s^+ \cap\Omega|} & \text{otherwise}.
                \end{cases}
        \end{align}
        If we select $\gamma = e$, then we have that 
        \begin{align}
            \alpha(\Omega,s,e)= \left(e \frac{|s^+ \cap\Omega|}{ {\rm rank}(H_{s,\Omega})}\right)^\frac{ {\rm rank}(H_{s,\Omega})}{|s^+ \cap\Omega|} \leq \max_{x \geq 0}\ (e x)^\frac{1}{x} = e.
            \label{eq:const_lemma1}
        \end{align}
        Now, if we substitute \eqref{eq:mutual_info_gaussian}, \eqref{eq:lemma_1} and \eqref{eq:const_lemma1} in \eqref{eq:cutset_1}, we get that
        \begin{align}
            \label{eq:cutset_final}
        \mathsf{C}_{\rm cs} &\leq \max_{\pr_{\{S_i\}}(\cdot)} \ \min_{\Omega \subseteq [1:N] \cup \{0\}} \left[\sum_{s} \lambda_s \log\det\left(I +P\widehat{H}_{s,\Omega}\widehat{H}^H_{s,\Omega}\right) + |\Omega|\log e\right]{+} 2\log(N{+}2) {+} N\log(\text{Card}(S_1))\nonumber \\
        &\leq \underbrace{\max_{\pr_{\{S_i\}}(\cdot)} \ \min_{\Omega \subseteq [1:N] \cup \{0\}} \sum_{s} \lambda_s \log\det\left(I +P\widehat{H}_{s,\Omega}\widehat{H}^H_{s,\Omega}\right)}_{\msf{C}_{\rm cs,iid}} + \underbrace{(N+1)\log e + 2\log(N{+}2) {+} N\log(\text{Card}(S_1))}_{\mathsf{GAP}}.
        \end{align}
The main observation in~\eqref{eq:cutset_final} is that an i.i.d Gaussian distribution on the inputs and a fixed schedule are within a constant additive gap from the information-theoretic cut-set upper bound on the capacity of the 1-2-1 network. With this, we can argue that $\msf{C}_{\rm cs,iid}$ is within a constant gap of the capacity. This is due to the fact that $\msf{C}_{\rm cs,iid}$ can be achieved using QMF as in~\cite{OzgurIT2013} or Noisy Network Coding as in~\cite{CardoneIT2014}.

Due to the special structure of the Gaussian 1-2-1 network we can further simplify $\msf{C}_{\rm cs,iid}$ by making use of the structure of $\widehat{H}_{s,\Omega}$ in \eqref{eq:cutset_final}. 
In particular, recall that, since every row (and column) in $\widehat{H}_{s,\Omega}$ has at most one non-zero element, then there exists a permutation matrix $\Pi_{s,\Omega}$ such that $\Pi_{s,\Omega}\widehat{H}_s$ is a diagonal matrix (not necessarily square). Thus we have 
        \begin{align}
            \label{eq:permuted_channel}
            \log\det\left(I +P\widehat{H}_{s,\Omega}\widehat{H}^H_{s,\Omega}\right) 
&\stackrel{(a)}=  \log\det\left(I +P \Pi_{s,\Omega}\widehat{H}_{s,\Omega}\widehat{H}^H_{s,\Omega}\Pi^T_{s,\Omega}\right)\nonumber\\
                              &\stackrel{(b)}= \sum_{i =1}^{\min\{|\Omega|,|\Omega^c|\}} \log\left(1 +P \left|[\Pi_{s,\Omega}\widehat{H}_{s,\Omega}]_{i,i}\right|^2\right),
        \end{align}
        where: (i) the equality in $(a)$ follows since permutation matrices are orthogonal matrices and thus multiplying by them only permutes the singular values of a matrix; (ii) the equality in $(b)$ follows since the permuted channel matrix $\Pi_{s,\Omega}\widehat{H}_{s,\Omega}$ can be represented as a parallel MIMO channel with $\min\{|\Omega|,|\Omega^c|\}$ active links. 
        We can rewrite the expression in~\eqref{eq:permuted_channel} as
        \begin{align}
            \log\det\left(I +P\widehat{H}_{s,\Omega}\widehat{H}^T_{s,\Omega}\right) &= \sum_{\substack{(i,j):\\ i \in s^+ \cap \Omega,\ j \in s^- \cap \Omega^c,\\ j \in s_{i,t},\ i \in s_{j,r}} } \log\left(1 +P\left| [\widehat{H}]_{j,i}\right|^2\right)\nonumber \\
            &= \sum_{\substack{(i,j):\\ i \in s^+ \cap \Omega,\ j \in s^- \cap \Omega^c,\\ j \in s_{i,t},\ i \in s_{j,r}} } \log\left(1 +P\left| h_{ji}\right|^2\right).
        \end{align}
Thus, by letting $\ell_{j,i} = \log\left(1 +P\left| h_{ji}\right|^2\right)$, we arrive at the following expression for $\msf{C}_{\rm cs,iid}$
       \begin{align}
            \label{eq:approx_to_flow}
            \msf{C}_{\rm cs,iid} &= \max_{\substack{\lambda: \|\lambda\|_1 = 1 \\ \lambda \geq 0}} \ \min_{\Omega \subseteq [1:N] \cup \{0\}} \sum_{s} \lambda_s \sum_{\substack{(i,j):\\ i \in s^+ \cap \Omega,\ j \in s^- \cap \Omega^c,\\ j \in s_{i,t},\ i \in s_{j,r}} } \ell_{j,i} \nonumber \\
            &=  \max_{\substack{\lambda: \|\lambda\|_1 = 1 \\ \lambda \geq 0}} \ \min_{\Omega \subseteq [1:N] \cup \{0\}} \sum_{s} \lambda_s \sum_{(i,j) \in [0:N+1]^2} \1_{\{j \in s_{i,t},\ i \in s_{j,r} \}}\1_{\{i \in \Omega,\ j \in \Omega^c\}} \ell_{j,i}\nonumber\\
            &=  \max_{\substack{\lambda: \|\lambda\|_1 = 1 \\ \lambda \geq 0}} \ \min_{\Omega \subseteq [1:N] \cup \{0\}}  \sum_{(i,j) \in [0:N+1]^2}\1_{\{i \in \Omega,\ j \in \Omega^c\}} \sum_{s} \lambda_s \1_{\{j \in s_{i,t},\ i \in s_{j,r} \}} \ell_{j,i}\nonumber\\
            &=  \max_{\substack{\lambda: \|\lambda\|_1 = 1 \\ \lambda \geq 0}} \ \min_{\Omega \subseteq [1:N] \cup \{0\}}  \sum_{\substack{(i,j): i \in \Omega,\\ j \in \Omega^c}} \left( \sum_{\substack{ s:\\ j \in s_{i,t},\\ i \in s_{j,r} }} \lambda_s \right) \ell_{j,i}\nonumber \\
            &=  \max_{\substack{\lambda: \|\lambda\|_1 = 1 \\ \lambda \geq 0}} \ \min_{\Omega \subseteq [1:N] \cup \{0\}}  \sum_{\substack{(i,j): i \in \Omega,\\ j \in \Omega^c}}\ell^{(s)}_{j,i},
        \end{align}
        where $\ell_{j,i}^{(s)}$ is defined as 
    \[
        \ell_{j,i}^{(s)} = \left( \sum_{\substack{ s:\\ j \in s_{i,t},\\ i \in s_{j,r} }} \lambda_s \right) \ell_{j,i}.
    \]
This concludes the proof that the capacity $\mathsf{C}$ of the Gaussian 1-2-1 network described in~\eqref{eq:model_2} can be characterized to within a constant gap as expressed in~\eqref{eq:constGap}.

\section{Gaussian FD Diamond 1-2-1 Network: Proof of Lemma~\ref{lem:diamond_1-2-1_ntwk} and Lemma~\ref{lem:simplification_diamond}(L1)}
\label{app:diamond_FD}
In this section, {we prove Lemma~\ref{lem:diamond_1-2-1_ntwk} for FD and Lemma~\ref{lem:simplification_diamond}(L1) by analyzing} the Gaussian FD 1-2-1 {FD} network with a diamond topology. In this network the source communicates with the destination by hopping through one layer of $N$ non-interfering relays. 
For this network the LP P1 in~\eqref{eq:CapPaths} can be further simplified by leveraging the two following implications of the sparse diamond topology:
\begin{enumerate}
\item In a Gaussian 1-2-1 diamond network, we have $N$ disjoint paths from the source to the destination, each passing through a different relay. We enumerate these paths with the index $i \in [1:N]$ depending on which relay is in the path. Moreover, each path $i \in [1:N]$ has a FD capacity equal to $\mathsf{C}_i = \min \left \{ \ell_{i,0}, \ell_{N+1,i} \right \}$;
\item {In the Gaussian FD 1-2-1 diamond network, each relay $i \in [1:N]$ appears in only one path from the source to the destination. Thus, when considering constraints $({\rm P1}b)$ and $({\rm P1}c)$ in~\eqref{eq:CapPaths} for $i \in [1:N]$ gives us that
\begin{align}
\label{eq:obs_2_FD_diamond}
    x_i \frac{\msf{C}_i}{\ell_{i,0}} \leq 1\quad \& \quad x_i \frac{\msf{C}_i}{\ell_{N+1,i}} \leq 1.
\end{align}
\item Note that $\mathsf{C}_i = \min \{ \ell_{i,0},\ell_{N+1,i} \}$. 
Therefore, one of the coefficients $\mathsf{C}_{i}/{\ell_{i, 0}}$ or $\mathsf{C}_{i}/{\ell_{N+1,i}}$ in \eqref{eq:obs_2_FD_diamond} is equal to $1$. This implies that a feasible solution of ${\rm P1}^d$, has $x_1 \leq 1$ and $x_2 \leq 1$. Therefore, the constraints $x_i \leq 1,\ \forall i \in[1:N]$, albeit redundant, can be added to the LP without reducing the feasibility region.
\item In the Gaussian FD 1-2-1 diamond network, the constraints due to the source and destination nodes, namely $({\rm P1}b)$ for $i=0$ and $({\rm P1}c)$ for $i=N+1$ in~\eqref{eq:CapPaths} gives us that
\begin{align}\label{eq:obs_3_FD_diamond}
\sum_{i \in [1:N]}  x_i \frac{\mathsf{C}_{i}}{\ell_{i, 0}} \leq  1,\qquad 
\sum_{i \in [1:N]}  x_i \frac{\mathsf{C}_{i}}{\ell_{N+1, i}} \leq  1.
\end{align}
Note that the constraints in~\eqref{eq:obs_3_FD_diamond} make the constraints \eqref{eq:obs_2_FD_diamond} redundant.
}
\end{enumerate}
By considering the two implications above, we can readily simplify P1 in~\eqref{eq:CapPaths} for Gaussian FD 1-2-1 networks with a diamond topology as follows 
{
\begin{align}
\label{eq:CapPathsDiam}
\begin{array}{llll}
{\rm P1}^d: &\msf{C}_{\rm cs,iid}  = {\rm max}  \displaystyle\sum_{i \in [1:N]} x_i \mathsf{C}_i   & & \\
&     ({\rm P}1a)^d \ 0 \leq x_i \leq 1 & \forall i \in [1:N], &  \\
&    ({\rm P}1b)^d \ \displaystyle\sum_{i \in [1:N]}  x_i \frac{\mathsf{C}_{i}}{\ell_{i, 0}} \leq  1, & & \\
&   ({\rm P}1c)^d \ \displaystyle\sum_{i \in [1:N]}  x_i \frac{\mathsf{C}_{i}}{\ell_{N+1, i}} \leq  1, &  &
\end{array}
\end{align}
which is the LP we have in Lemma~\ref{lem:diamond_1-2-1_ntwk}.}

{To prove Lemma~\ref{lem:simplification_diamond}(L1), we} observe that for a bounded LP, there always exists an optimal corner point. Furthermore, at any corner point in the LP ${\rm P1}^d$, we have at least $N$ constraints satisfied with equality among $(1a)^d$, $(1b)^d$ and $(1c)^d$. Therefore, we have at least $N-2$ constraints in $(1a)^d$ satisfied with equality that make  {linearly independent equations} (since $(1b)^d$ and $(1c)^d$ combined represent only two constraints). {Furthermore, recall that as mentioned earlier all constraints $x_i \leq 1$ are redundant.} Thus, at least $N-2$ relays are turned off (i.e., $x_i = 0$), i.e., at most two relays are sufficient to characterize the approximate capacity of any $N$-relay Gaussian FD 1-2-1 network with a diamond topology.
{
\section{Gaussian HD Diamond 1-2-1 Network: Proof of Lemma~\ref{lem:diamond_1-2-1_ntwk} and Lemma~\ref{lem:simplification_diamond}(L2)}
\label{app:diamond_HD}
We prove Lemma~\ref{lem:diamond_1-2-1_ntwk} for HD diamond networks in the first subsection and later prove Lemma~\ref{lem:simplification_diamond}(L2) in the following subsection.
\subsection{Proof of Lemma~\ref{lem:diamond_1-2-1_ntwk} for an HD diamond network}

Throughout this section, we slightly abuse notation by defining $\ell_{i} = \ell_{i,0}$ and $r_i = \ell_{N+1,i}$. Based on this definition, we can write the approximate capacity expression \eqref{eq:apprCap} as
\begin{align}
\msf{C}_{\rm cs,iid} &= \max_{\substack{\lambda_s: \lambda_s \geq 0 \\ \sum_s \lambda_s = 1}} \min_{\Omega \subseteq [1:N] \cup \{0\}} \! \sum_{i \in \Omega^c}\left( \sum_{\substack{ s:\\ i \in s_{0,t},\\ 0 \in s_{i,r} }} \lambda_s \right) \ell_{i} + \sum_{i \in \Omega}\left( \sum_{\substack{ s:\\ (N+1) \in s_{i,t},\\ i \in s_{N+1,r} }} \lambda_s \right) r_{i} \nonumber \\
&= \max_{\substack{\lambda_s: \lambda_s \geq 0 \\ \sum_s \lambda_s = 1}} \min_{\Omega \subseteq [1:N] \cup \{0\}} \! \sum_{i=1}^N \left[\1_{\{i \in \Omega^c\}} \left( \sum_{\substack{ s:\\ i \in s_{0,t},\\ 0 \in s_{i,r} }} \lambda_s \right) \ell_{i} + \1_{\{i \in \Omega\}} \left( \sum_{\substack{ s:\\ (N+1) \in s_{i,t},\\ i \in s_{N+1,r} }} \lambda_s \right) r_{i}\right]\nonumber \\
&= \max_{\substack{\lambda_s: \lambda_s \geq 0 \\ \sum_s \lambda_s = 1}} \min_{\Omega \subseteq [1:N] \cup \{0\}} \! \sum_{i=1}^N \min \left\{\left( \sum_{\substack{ s:\\ i \in s_{0,t},\\ 0 \in s_{i,r} }} \lambda_s \right) \ell_{i},\ \left( \sum_{\substack{ s:\\ (N+1) \in s_{i,t},\\ i \in s_{N+1,r} }} \lambda_s \right) r_{i}\right\}. \label{eq:approx_cap_diamond_rewritten}
\end{align}

Our first directive is to show that the approximate capacity $\msf{C}_{\rm cs,iid}$ in \eqref{eq:approx_cap_diamond_rewritten} is equivalent to solving the LP P4
\begin{align}
    \label{eq:P4}
    {\rm P4 :}\quad {\rm maximize}\quad& \sum_{i = 1}^N \lambda_{\ell_i} \ell_i \nonumber \\
    {\rm subject\ to}\quad&  ({\rm P4}a)\ \lambda_{\ell_i} \ell_i = \lambda_{r_i} r_i \qquad& \forall i \in [1:N], \nonumber \\
     & ({\rm P4}b)\ \sum_{i=1}^N \lambda_{\ell_i} \leq 1,\quad \sum_{i=1}^N \lambda_{r_i} \leq 1, \\
     & ({\rm P4}c)\ \lambda_{\ell_i} + \lambda_{r_i} \leq 1\qquad &\forall i \in [1:N], \nonumber
\end{align}
where: (i) $f_i = \lambda_{\ell_i} \ell_i = \lambda_{r_i} r_i$ represents the data flow through the $i$-th relay; (ii) $\lambda_{\ell_i}$ (respectively, $\lambda_{r_i}$) represents the fraction of time in which the link from the source to relay $i$ (respectively, from relay $i$ to the destination) is active. 
Note that since the network is operating in HD, then in \eqref{eq:approx_cap_diamond_rewritten}, we have that $|s_{i,t}| + |s_{i,r}| \leq 1,\ \forall i \in [1:N]$ which is captured by the constraint $({\rm P}4c)$ above.

To show the first direction (i.e., a feasible schedule in \eqref{eq:approx_cap_diamond_rewritten} gives a feasible point in the LP P1), we define the following transformation
\begin{align}
    \forall i \in [1:N] \ \ : \quad &f_i = \min\left\{\left(\sum_{\substack{ s:\\ i \in s_{0,t},\\ 0 \in s_{i,r} }} \lambda_s\right)\ell_i,\ \left(\sum_{\substack{ s:\\ (N+1) \in s_{i,t},\\ i \in s_{N+1,r} }} \lambda_s\right)r_i\right\}\nonumber\\
    &\qquad\qquad \lambda_{\ell_i} = \frac{f_i}{\ell_i},\quad \lambda_{r_i} = \frac{f_i}{r_i}.
    \label{eq:cs_2_P4}
\end{align}
Using this transformation, we have that
\begin{align*}
    \lambda_{\ell_i} \ell_i &= f_i = \lambda_{r_i} r_i,\quad \forall i \in [1:N] &\implies ({\rm P}4a) \\
    \sum_{i = 1}^N \lambda_{\ell_i} &= \sum_{i = 1}^N \frac{f_i}{\ell_i} = \sum_{i = 1}^N \frac{1}{\ell_i} \min\left\{\left(\sum_{\substack{ s:\\ i \in s_{0,t},\\ 0 \in s_{i,r} }} \lambda_s\right)\ell_i,\ \left(\sum_{\substack{ s:\\ (N+1) \in s_{i,t},\\ i \in s_{N+1,r} }} \lambda_s\right)r_i\right\} \\
    &\qquad \qquad \leq \sum_{i=1}^N \frac{1}{\ell_i}\left(\sum_{s} \lambda_s\right) \min\{\ell_i,r_i\}\leq \frac{\min\{\ell_i,r_i\}}{\ell_i} \leq 1&\implies ({\rm P}4b)\\
        \sum_{i = 1}^N \lambda_{r_i} &= \sum_{i = 1}^N \frac{f_i}{r_i} = \sum_{i = 1}^N \frac{1}{r_i} \min\left\{\left(\sum_{\substack{ s:\\ j \in s_{i,t},\\ i \in s_{j,r} }} \lambda_s\right)\ell_i,\ \left(\sum_{\substack{ s:\\ (N+1) \in s_{i,t},\\ i \in s_{N+1,r} }} \lambda_s\right)r_i\right\} \\
    &\qquad \qquad \leq \sum_{i=1}^N \frac{1}{r_i}\left(\sum_{s} \lambda_s\right) \min\{\ell_i,r_i\}\leq \frac{\min\{\ell_i,r_i\}}{r_i} \leq 1&\implies ({\rm P}4b)\\
    \lambda_{\ell_i} {+} \lambda_{r_i} &= \left[\frac{1}{\ell_i} + \frac{1}{r_i}\right] \min\left\{\left(\sum_{\substack{ s:\\ i \in s_{0,t},\\ 0 \in s_{i,r} }} \lambda_s\right)\ell_i,\ \left(\sum_{\substack{ s:\\ (N+1) \in s_{i,t},\\ i \in s_{N+1,r} }} \lambda_s\right)r_i\right\}\\
    &{=} {\min}\left\{\left(\sum_{\substack{ s:\\ i \in s_{0,t},\\ 0 \in s_{i,r} }} \lambda_s\right), \left(\sum_{\substack{ s:\\ (N{+}1) {\in} s_{i,t},\\ i \in s_{N{+}1,r} }} \lambda_s\right)\frac{r_i}{\ell_i}\right\} {+} {\min}\left\{\left(\sum_{\substack{ s:\\ i \in s_{0,t},\\ 0 \in s_{i,r} }} \lambda_s\right)\frac{\ell_i}{r_i}, \left(\sum_{\substack{ s:\\ (N{+}1) \in s_{i,t},\\ i \in s_{N{+}1,r} }} \lambda_s\right)\right\}\\
    &\leq \left(\sum_{\substack{ s:\\ i \in s_{0,t},\\ 0 \in s_{i,r} }} \lambda_s\right) + \left(\sum_{\substack{ s:\\ (N{+}1) \in s_{i,t},\\ i \in s_{N{+}1,r} }} \lambda_s\right) \leq \sum_s \lambda_s = 1&\implies ({\rm P}4c).
\end{align*}
Thus, a feasible schedule in~\eqref{eq:approx_cap_diamond_rewritten} gives a feasible point in the LP P4 in~\eqref{eq:P4}.
Furthermore, by substituting~\eqref{eq:cs_2_P4} in~\eqref{eq:approx_cap_diamond_rewritten}, we get that the rate achieved by the schedule is 
\[
\sum_{i=1}^N \min\left\{ \left(\sum_{\substack{ s:\\ i \in s_{0,t},\\ 0 \in s_{i,r} }} \lambda_s\right)  \ell_i \ ,\ \left(\sum_{\substack{ s:\\ (N{+}1) \in s_{i,t},\\ i \in s_{N{+}1,r} }} \lambda_s\right) r_i \right\} = \sum_{i=1}^N f_i = \sum_{i=1}^N \lambda_{\ell_i}\ell_i
\]
which is equal to the objective function value of P4 in~\eqref{eq:P4}.

To prove the opposite direction (i.e., P4 $\to$ \eqref{eq:approx_cap_diamond_rewritten}), we show that we can map an optimal solution in P4 to a feasible point (schedule) in \eqref{eq:approx_cap_diamond_rewritten} with a rate equal to the optimal value of P4.
First, note that the LP P4 has $2N$ variables. 
As a result, a corner point in P4, should have at least $N$ constraints from $({\rm P}4b)$, $({\rm P}4c)$ and $({\rm P}4d)$ satisfied with equality (we already have $N$ other equality constraints due to $({\rm P}4a)$.
We now prove an interesting property about optimal corner points in P4 which facilitates our proof.
\begin{prope}
    \label{prop:property_optimal_1}
    For any optimal corner point $\{\lambda_{\ell_i}^\star,\lambda_{r_i}^\star\}$ in P4, there exists an $i' \in [1:N]$ such that $\lambda^\star_{\ell_{i'}} + \lambda^\star_{r_{i'}} =1$.
\end{prope}
\begin{proof}
    To prove this property, we are going to consider three cases depending on which constraints are satisfied with equality at an optimal corner point.

\noindent {\bf 1) Both conditions in $({\rm P}4b)$ are not satisfied with equality:} In this case, a corner point has at least $N$ constraints among $({\rm P}4c)$ and $({\rm P}4d)$ satisfied with equality.
            It is not difficult to see that, in order for the corner point to be optimal, at least for one $i'$ we have $\lambda^\star_{\ell_{i'}} + \lambda^\star_{r_{i'}} =1$, otherwise we have a non-optimal zero value for the objective function.

\smallskip

        \noindent{\bf 2) Only one condition in $({\rm P4}b)$ is not satisfied with equality:} In this case, a corner point has at least $N-1$ constraints among $({\rm P4}c)$ and $({\rm P4}d)$ satisfied with equality. 
            Thus, there exists at most one $i$ such that $0 < \lambda_{\ell_i} + \lambda_{r_i} < 1$.
            Additionally, by adding the conditions in $({\rm P4}b)$, we get that
            \begin{align}
                1 < \sum_{i = 1}^N \left( \lambda_{\ell_i} + \lambda_{r_i} \right) < 2.
                \label{eq:prope_1_2}
            \end{align}
            Thus, there exists one $i'$ such that $\lambda_{\ell_{i'}} + \lambda_{r_{i'}} = 1$, otherwise, we cannot satisfy the lower bound in \eqref{eq:prope_1_2}.

\smallskip

       \noindent {\bf 3) Both conditions in $({\rm P4}b)$ are satisfied with equality:} In this case, a corner point has at least $N-2$ constraints among $({\rm P}4c)$ and $({\rm P}4d)$ satisfied with equality.
            Thus, there exist at most two $i$ such that $0 < \lambda_{\ell_i} + \lambda_{r_i} < 1$.
            Furthermore, adding the constraints in $({\rm P4}b)$ implies that 
            \begin{align}
                \sum_{i = 1}^N \left( \lambda_{\ell_i} + \lambda_{r_i} \right) = 2.
                \label{eq:prope_1_1}
            \end{align}
        The two aforementioned observations imply that all $N-2$ equalities cannot be from $({\rm P}4d)$, otherwise we have that $\sum_{i = 1}^N \left( \lambda_{\ell_i} + \lambda_{r_i} \right) < 2$, which contradicts \eqref{eq:prope_1_1}. Thus, there exists a constraint in $({\rm P}4c)$ that is satisfied with equality, i.e., $\lambda^\star_{\ell_{i'}} + \lambda^\star_{r_{i'}} =1$ for some $i' \in [1:N]$.
\end{proof}

We now use Property~\ref{prop:property_optimal_1} to show that, for any optimal point in P4, we can find a feasible schedule in \eqref{eq:approx_cap_diamond_rewritten} that gives a rate equal to the objective function in P4. For an optimal point $(\lambda_{\ell_1}^\star,\lambda_{r_1}^\star,\dots,\lambda_{\ell_N}^\star,\lambda_{r_N}^\star)$, let $i'$ be the index such that $\lambda^\star_{\ell_{i'}} + \lambda^\star_{r_{i'}} = 1$ (such an index exists thanks to Property~\ref{prop:property_optimal_1}). Thus, we have the following condition for our optimal point
\begin{align}
    \label{eq:properties_optimal_P2}
    \lambda^\star_{\ell_{i'}} + \lambda^\star_{r_{i'}} &= 1,\nonumber \\
    \sum_{i \in [1:N]\backslash \{i'\}} \lambda^\star_{\ell_i} &\leq 1 - \lambda^\star_{\ell_{i'}} =\lambda^\star_{r_{i'}} ,\\
    \sum_{i\in [1:N]\backslash \{i'\}} \lambda^\star_{r_i} &\leq 1 - \lambda^\star_{r_{i'}} = \lambda^\star_{\ell_{i'}}.\nonumber
\end{align}

Note that, any state $s$ in the 1-2-1 Gaussian HD diamond network activates at most two links in the network: a link between the source and the $m$-th relay and/or well as the link between the $n$-th relay and the destination. For brevity, in our construction we will denote with $s_{m,n}$ the state that activates the link from the source to the $i$-th relay in the diamond network as well as the links from the $n$-th relay to the destination (If either link is not activated, the corresponding index is $\emptyset$). We also use $\lambda_{s_{m,n}}$ to denote the fraction of time during which this network state is active. Using this notation, we can construct the following schedule from the given optimal point in P4
\begin{align}
    \label{eq:definitions_P2_to_approx_cap}
    \lambda_{s_{i,i'}} &= \lambda^\star_{\ell_i},& \forall i \in [1:N]\backslash\{i'\},\nonumber \\
    \lambda_{s_{\emptyset,i'}} &= \lambda^\star_{r_{i'}} - \sum_{i \in [1:N]\backslash\{i'\}}\lambda^\star_{\ell_i},&\nonumber\\
    \lambda_{s_{i',i}} &= \lambda^\star_{r_i},&\forall i \in [1:N]\backslash\{i'\},\\
    \lambda_{s_{i',\emptyset}} &= \lambda^\star_{\ell_{i'}} - \sum_{i \in [1:N]\backslash\{i'\}}\lambda^\star_{r_i}.&\nonumber
\end{align}
The activation time of all other states, except those described above, is set to zero. 

From~\eqref{eq:properties_optimal_P2}, we know that all values defined in~\eqref{eq:definitions_P2_to_approx_cap} are positive.
We can verify that the generated schedule is feasible, i.e., the sum of all $\lambda$ has to add up to one as follows
\begin{align*}
    \sum_s \lambda_s &= \lambda_{s_{0,i'}} + \sum_{i \in [1:N]\backslash\{i'\}} \lambda_{s_{i,i'}} +\sum_{i \in [1:N]\backslash\{i'\}} \lambda_{s_{i',i}} + \lambda_{s_{i',0}} \stackrel{(a)}= \lambda^\star_{\ell_{i'}} + \lambda^\star_{r_{i'}}  \stackrel{(b)}= 1,
\end{align*}
where: (i) the equality in $(a)$ follows from the definitions in~\eqref{eq:definitions_P2_to_approx_cap} and (ii) the equality in $(b)$ follows from Property~\ref{prop:property_optimal_1}.

In order to conclude the mapping from P4 to \eqref{eq:approx_cap_diamond_rewritten}, we need to verify that the rate achieved with the constructed schedule in~\eqref{eq:definitions_P2_to_approx_cap} is equal to the optimal value of the LP P4.
From~\eqref{eq:approx_cap_diamond_rewritten}, we get that
\begin{align}
    \label{eq:P2_in_approx_cap}
    &\sum_{i=1}^N\min \left\{\left( \sum_{\substack{ s:\\ i \in s_{0,t}, 0 \in s_{i,r} }} \lambda_s \right) \ell_{i},\ \left( \sum_{\substack{ s:\\ (N+1) \in s_{i,t}, i \in s_{N+1,r} }} \lambda_s \right) r_{i}\right\}\nonumber \\
    &=\sum_{i\in [1:N]\backslash\{i'\}} \min \left\{\left( \sum_{\substack{ s:\\ i \in s_{0,t},\\ 0 \in s_{i,r} }} \lambda_s \right) \ell_{i},\ \left( \sum_{\substack{ s:\\ (N+1) \in s_{i,t},\\ i \in s_{N+1,r} }} \lambda_s \right) r_{i}\right\}\nonumber + \min \left\{\left( \sum_{\substack{ s:\\ i \in s_{0,t},\\ 0 \in s_{i,r} }} \lambda_s \right) \ell_{i},\ \left( \sum_{\substack{ s:\\ (N+1) \in s_{i,t},\\ i \in s_{N+1,r} }} \lambda_s \right) r_{i}\right\}\nonumber \\
&=\left(\sum_{i\in [1:N]\backslash\{i'\}} \min\left\{ \lambda_{s_{i,i'}}  \ell_i \ ,\ \lambda_{s_{i',i}} r_i \right\}\right) + \min\left\{ \left(\sum_{j \in [0:N]\backslash\{i'\}} \lambda_{s_{i',j}}\right)  \ell_{i'} \ ,\ \left(\sum_{k \in [0:N]\backslash\{i'\}} \lambda_{s_{k,i'}}\right) r_{i'} \right\}\nonumber \\
&=\left(\sum_{i\in [1:N]\backslash\{i'\}} \min\left\{ \lambda^\star_{\ell_{i}}  \ell_i \ ,\ \lambda^\star_{r_{i}} r_i \right\}\right) + \min\left\{\lambda_{\ell_{i'}}\ell_{i'} \ ,\ \lambda_{r_{i'}} r_{i'} \right\} =\sum_{i\in [1:N]} \min\left\{ \lambda^\star_{\ell_{i}}  \ell_i \ ,\ \lambda^\star_{r_{i}} r_i \right\}.
\end{align}

Now note that, since the optimal corner point in P4 is feasible in P4, then $\lambda^\star_{\ell_{i}}  \ell_i = \lambda^\star_{r_{i}} r_i$, $\forall i \in [1:N]$.
Thus the expression in~\eqref{eq:P2_in_approx_cap} can be rewritten as $\sum_{i=1}^N \lambda^\star_{\ell_i} \ell_i$, which is the optimal objective function value in P4.
Thus, we can now conclude that \eqref{eq:approx_cap_diamond_rewritten} is equivalent to P4.

We are now going to relate the LP P4 discussed above to the LP in Lemma~\ref{lem:diamond_1-2-1_ntwk}.
Recall that for a two hop Half-Duplex path with link capacities $\ell_i$ and $r_i$, the capacity is given by 
\[
\msf{C}_i = \frac{\ell_i r_i}{\ell_i + r_i}.
\]
Thus, we ca write the LP ${\rm P1^d}$ as 
\begin{align}
     \label{eq:P3}
    {\rm P1^d :}\quad {\rm maximize}\quad& \sum_{i = 1}^N x_i \frac{\ell_i r_i}{\ell_i + r_i}\nonumber \\
    {\rm subject\ to}\quad&  ({\rm P1}a)^{\rm d}\ 0 \leq x_i \leq 1 \qquad\qquad \forall i \in [1:N], \nonumber \\
    & ({\rm P1}b)^{\rm d}\ \sum_{i=1}^N x_i \frac{\ell_i}{\ell_i + r_i} \leq 1,\\
     & ({\rm P1}c)^{\rm d}\ \sum_{i=1}^N x_i \frac{r_i}{\ell_i + r_i} \leq 1. \nonumber
\end{align}
We are now going to show that P1 is equivalent to the LP P4 and, as a consequence, it is to the formulation of $\msf{C}_{\rm cs,iid}$ in~\eqref{eq:approx_cap_diamond_rewritten}.
To do this, we are going to show how a feasible point in P4 can be transformed into a feasible point in ${\rm P1^d}$ and vice versa. 
\begin{enumerate}
    \item \underline{P4 $\to$ ${\rm P1^d}$.} Define $x_i$ to be
        \begin{align}
            x_i = \lambda_{\ell_i} \frac{\ell_i + r_i}{r_i},\quad \forall i \in [1:N].
            \label{eq:P2_2_P3}
        \end{align}
        Using this transformation, we get that the constraints in P4 imply the following
        \begin{align*}
            ({\rm P}4b)\ &: \quad 1 \geq \sum_{i=1}^N\lambda_{\ell_i} = \sum_{i=1}^N  x_i \frac{r_i}{\ell_i + r_i}& \implies ({\rm P1}c)^{\rm d}\\
            ({\rm P}4b)\ &: \quad 1 \geq \sum_{i=1}^N\lambda_{r_i} \stackrel{({\rm P}4a)}= \sum_{i=1}^N\lambda_{\ell_i}\frac{\ell_i}{r_i} = \sum_{i=1}^N  x_i \frac{\ell_i}{\ell_i + r_i} &\implies ({\rm P1}b)^{\rm d}\\
            ({\rm P}4c)\ &: \quad \forall i \in [1:N], \quad 1 \geq  \lambda_{\ell_i} + \lambda_{r_i} \stackrel{({\rm P}4a)}= \lambda_{\ell_i} \left(1 + \frac{\ell_i}{r_i} \right) = x_i \frac{r_i}{\ell_i + r_i}\left(1+\frac{\ell_i}{r_i}\right) = x_i & \implies ({\rm P1}a)^{\rm d}\\
            ({\rm P}4d)\ &: \quad \forall i \in [1:N], \quad 0 \leq \lambda_{\ell_i} \frac{\ell_i + r_i}{r_i} = x_i &\implies  ({\rm P1}a)^{\rm d}
        \end{align*}
        \begin{align*}
            (\text{P4 objective function})\ &: \sum_{i=1}^N \lambda_{\ell_i}\ell_i =  \sum_{i=1}^N x_i \frac{r_i}{\ell_i + r_i}\ell_i = (\text{${\rm P1^d}$ objective function}).
        \end{align*}

        Thus for any feasible point in P4, we get a feasible point in ${\rm P1^d}$ using the transformation in \eqref{eq:P2_2_P3} that has the same objective function with the same value as the original point in P4.
    \item \underline{${\rm P1^d}$ $\to$ P4.} Define $\lambda_{\ell_i}$ and $\lambda_{r_i}$ to be
        \begin{align}
            \lambda_{\ell_i} = x_i \frac{r_i}{\ell_i + r_i},\quad \lambda_{r_i} = x_i \frac{\ell_i}{\ell_i + r_i},\quad \forall i \in [1:N].
            \label{eq:P3_2_P2}
        \end{align}
        Note that the transformation above directly implies condition ({\rm P}4a) in P4. Now, we are going to show that the constraints in ${\rm P1^d}$ when applied to~\eqref{eq:P3_2_P2} imply the rest of the constraints in P4 as follows
        \begin{align*}
            ({\rm P1}a)^{\rm d}\ &: \quad 1 \geq x_i = x_i \left( \frac{r_i}{\ell_i + r_i} + \frac{\ell_i}{\ell_i + r_i}\right)  = \lambda_{\ell_i} + \lambda_{r_i} &\implies ({\rm P}4c)\\
            ({\rm P1}a)^{\rm d}\ &: \quad 0 \leq x_i \frac{r_i}{r_i + \ell_i} = \lambda_{\ell_i} &\implies ({\rm P}4d)\\
            ({\rm P1}b)^{\rm d}\ &: \quad 1 \geq \sum_{i=1}^N  x_i \frac{\ell_i}{\ell_i + r_i} = \sum_{i=1}^N\lambda_{r_i} &\implies ({\rm P}4b)\\
            ({\rm P1}c)^{\rm d}\ &: \quad 1 \geq \sum_{i=1}^N  x_i \frac{r_i}{\ell_i + r_i} = \sum_{i=1}^N\lambda_{\ell_i} &\implies ({\rm P}4b)
        \end{align*}
        \begin{align*}
            (\text{${\rm P1^d}$ objective function})\ &: \sum_{i=1}^N x_i \frac{r_i}{\ell_i + r_i}\ell_i = \sum_{i=1}^N \lambda_{\ell_i}\ell_i = (\text{P4 objective function}).
        \end{align*}
\end{enumerate}

Thus the two problems ${\rm P1^d}$ and P4 are equivalent. This concludes the proof of Lemma~\ref{lem:diamond_1-2-1_ntwk} for the HD case.

\subsection{Proof of Lemma~\ref{lem:simplification_diamond}(L2) for an HD diamond network}
  We first prove the following property of the optimal corner points in the LP ${\rm P1^d}$ in~\eqref{eq:P3}.

    \begin{prope}
        If we have a 1-2-1 Gaussian HD diamond network, then for any optimal corner point solution of ${\rm P1^d}$, at least one of the constraints in $({\rm P}1b)^{\rm d}$ and $({\rm P}1c)^{\rm d}$ is satisfied with equality.
    \end{prope}
    \begin{proof}
        We are going to prove Property~2 by contradiction.
        Note that, since the LP ${\rm P1^d}$ has $N$ variables, then any corner point in ${\rm P1^d}$ has at least $N$ constraints satisfied with equality. 
        Now, assume that we have an optimal point $(x_1^\star,x_2^\star, \dots, x_N^\star)$ such that neither $({\rm P}1b)^{\rm d}$ nor $({\rm P}1c)^{\rm d}$ is satisfied with equality. This implies that the constraints satisfied with equality are only of the type $({\rm P}1a)^{\rm d}$. Thus, from the constraints in $({\rm P}1a)^{\rm d}$, we have that $x_i^\star \in \{0,1\}, \forall i \in [1:N]$.
        Additionally, $({\rm P}1b)^{\rm d}$ and $({\rm P}1c)^{\rm d}$ being strict inequalities implies that $\sum_{i=1}^N x^\star_i < 2$. 
Thus, there exists at most one $i'$, such that $x_{i'}^\star = 1$, while $x_j^\star = 0, \forall j \in [1:N]\backslash\{i'\}$.

            Now, if we pick some $k \neq i'$ and set $x_k^\star = \varepsilon  > 0$ such that both $({\rm P}1b)^{\rm d}$ and $({\rm P}1c)^{\rm d}$ are still satisfied, then we increase the objective function by $\varepsilon \frac{\ell_k r_k}{\ell_k + r_k}$, which contradicts the fact that $(x_1^\star,x_2^\star, \dots, x_N^\star)$ is an optimal solution.
    \end{proof}

    Now using Property~2, we are going to prove Lemma~\ref{lem:simplification_diamond}(L2) by considering the following two cases: (i) There exists an optimal corner point for which only one of the constraints in $({\rm P}1b)^{\rm d}$ and $({\rm P}1c)^{\rm d}$ is satisfied with equality, and (ii) all optimal corner points have both $({\rm P}1b)^{\rm d}$ and $({\rm P}1c)^{\rm d}$ satisfied with equality.

    \begin{enumerate}
        \item {\bf An optimal corner point exists with only one among $({\rm P}1b)^{\rm d}$ and $({\rm P}1c)^{\rm d}$ satisfied with equality.}
            We denote this optimal corner point as $(x_1^\star,x_2^\star, \dots, x_N^\star)$. 
            Since only one among $({\rm P}1b)^{\rm d}$ and $({\rm P}1c)^{\rm d}$ is satisfied with equality, then at least $N-1$ constraints of the type $({\rm P}1a)^{\rm d}$ are satisfied with equality.
            Also note that, since only one among $({\rm P}1b)^{\rm d}$ and $({\rm P}1c)^{\rm d}$ is satisfied with equality, then this implies that $\sum_{i = 1}^N x_i < 2$.
            This implies that, although we have at least $N-1$ constraints in $({\rm P}1a)^{\rm d}$ satisfied with equality, we have at most one $i'$ such that $x_{i'}^\star = 1$.
            As a result, at least $N-2$ of the constraints satisfied with equality from $({\rm P}1a)^{\rm d}$ are of the form $x_i = 0$.
            This proves that at least $N-2$ relays are not utilized at this optimal corner point, which proves Lemma~\ref{lem:simplification_diamond}(L2) in this case.

        \item {\bf All optimal corner points have $({\rm P}1b)^{\rm d}$ and $({\rm P}1c)^{\rm d}$ satisfied with equality.}
            Pick an optimal corner point and denote it as $(x_1^\star,x_2^\star, \dots, x_N^\star)$. 
        Define $\mcal{F}^\star_x = \{i | 0 < x^\star_i < 1\}$ and $\mcal{I}_x^\star = \{i | x^\star_i  = 1\}$, i.e., the sets of indices of the variables with non-integer and unitary values, respectively.
        The fact that both $({\rm P}1b)^{\rm d}$ and $({\rm P}1c)^{\rm d}$ are satisfied with equality implies that $\sum_{i=1}^N x^\star_i = 2$, which implies that $|\mcal{I}_x^\star| \leq 2$.
            Additionally, since we are considering a corner point, then we have that at least $N-2$ constraints of the type $({\rm P}1a)^{\rm d}$ are satisfied with equality.
            This implies that $|\mcal{F}_x^\star| \leq 2$. Note that, if $|\mcal{F}_x^\star| + |\mcal{I}_x^\star| \leq 3$ for all optimal corner points, then we have proved Lemma~\ref{lem:simplification_diamond}(L2) for this case. Thus, we now show that the events $\{|\mcal{F}_x^\star| = 2\}$ and $\{|\mcal{I}_x^\star|= 2\}$ are mutually exclusive (i.e., disprove the possibility that $|\mcal{F}_x^\star| + |\mcal{I}_x^\star| = 4$). This follows by observing the following relation
            \begin{align*}
                2 = \sum_{i = 1}^N x^\star_i = \sum_{i \in [1:N]\backslash\mcal{I}_x^\star} x^\star_i + \sum_{i \in \mcal{I}_x^\star} x^\star_i = \sum_{i \in [1:N]\backslash\mcal{I}_x^\star} x^\star_i + |\mcal{I}_x^\star|.
            \end{align*}
            Thus
            \begin{align*}
                |\mcal{I}_x^\star| = 2 \implies  \sum_{i \in [1:N]\backslash\mcal{I}_x^\star} x^\star_i  = 0 \implies |\mcal{F}_x^\star| = 0,
            \end{align*}
            which proves that the two events are mutually exclusive. This concludes the proof of Lemma~\ref{lem:simplification_diamond}(L2).
    \end{enumerate}

\section{Proof of Lemma~\ref{lem:simplification_diamond}(L3)}
\label{app:Lemma5L3}
The proof of Lemma~\ref{lem:simplification_diamond}(L3) for the FD case follows directly from~\ref{lem:simplification_diamond}(L1) by taking only the two paths (relays) needed to achieve the $\msf{C}_{\rm cs,iid}$. Without loss generality, we assume that relays 1 and 2 are the relays in question. Then we have using the optimal fractions $x_1^\star$ and $x_2^\star$ that
\[
    \msf{C}_{\rm cs,iid} = x_1^\star \msf{C}_1 + x_2^\star \msf{C}_2 \stackrel{({\rm P1}a)^{\rm d}}\leq \msf{C}_1 + \msf{C}_2,
\]
which proves that either $\msf{C}_1$ or $\msf{C}_2$ are greater than or equal half $\msf{C}_{\rm cs,iid}$.

To prove Lemma~\ref{lem:simplification_diamond}(L3) for the HD case, note that for a HD network $\msf{C}_i$ in ${\rm P1}^{\rm d}$ is given by 
\[
\msf{C}_i = \frac{\ell_{i,0}\ \ell_{N+1,i}}{\ell_{i,0}\ +\ \ell_{N+1,i}}.
 \]
Thus, by adding the constraints $({\rm P1}b)^{\rm d}$ and $({\rm P1}c)^{\rm d}$, we have the following implication for any feasible point in ${\rm P1}^{\rm d}$
\begin{align}
2 &\geq \sum_{i \in [1:N]}  x_i \frac{\mathsf{C}_{i}}{\ell_{i, 0}} + \sum_{i \in [1:N]}  x_i \frac{\mathsf{C}_{i}}{\ell_{N+1,i}} \nonumber \\
 &= \sum_{i \in [1:N]}  x_i \frac{\ell_{N+1,i}}{\ell_{i,0}+\ell_{N+1,i}} + \sum_{i \in [1:N]}  x_i \frac{\ell_{i,0}}{\ell_{i,0}+\ell_{N+1,i}} = \sum_{i \in [1:N]}  x_i.
\end{align}
Now, assume without loss of generality that the path through relay 1 has the largest HD approximate capacity.
Then, for any optimal point $x_i^\star$ that solves ${\rm P1}^{\rm d}$ in the HD case, we have
\[
\msf{C}_{\rm cs,iid} = \sum_{i \in [1:N]} x^\star_i \mathsf{C}_i \leq \left(\sum_{i \in [1:N]} x^\star_i \right) \msf{C}_1 \leq 2 \msf{C}_1.
\]
This proves that the approximate capacity of the best path in the network is at least half the of $\msf{C}_{\rm cs,iid}$.
}

\section{Equivalence between P3 and P1}
\label{app:P3EqualP1}
In this section, we prove the equivalence between the LPs P1 and P3 (which, as proved in Section~\ref{sec:ProofMainTh} is equivalent to P2), hence concluding the proof of Theorem~\ref{thm:MainTh}. In particular, our proof consists of two steps.

We first show that the LP in P3 is equivalent to the LP P5 below
\begin{align}
     \label{eq:P1}
\begin{array}{llll}
{\rm P5 :} &\msf{C}_{\rm cs,iid}= {\rm max}  \sum_{p \in \mcal{P}} F_p & & \\
& ({\rm P}5a) \ F_p \geq 0 & \forall p \in \mcal{P}, &\\
& ({\rm P}5b) \  F_p = \lambda^p_{\ell_{p\pnext(i),i}}  \ell_{p\pnext(i),i} & \forall i \in p \backslash \{N+1\}, \forall p \in \mcal{P}, &\\
&  ({\rm P}5c)\ F_p = \lambda^p_{\ell_{i, p\pprev(i)}}  \ell_{i,p\pprev(i)} & \forall i \in p \backslash \{0\}, \forall p \in \mcal{P}, & \\
& ({\rm P}5d) \ \sum_{p \in \mcal{P}_i} \lambda^p_{\ell_{p\pnext(i),i}} \leq 1 & \forall i \in [0:N], & \\
& ({\rm P}5e)\ \sum_{p \in \mcal{P}_i} \lambda^p_{\ell_{i,p\pprev(i)}} \leq 1 & \forall i \in [1:N+1], &
\end{array}
\end{align}
and then show that P5 is equivalent to P1.

\noindent \underline{P3 $\to$ P5.} For $(i,j) \in p$ such that $j = p\pnext(i)$, define the variable $\lambda_{\ell_{j,i}}^p$ to be
\begin{align}
\lambda_{\ell_{j,i}}^p = \frac{F_p}{\ell_{j,i}}.
    \label{eq:P3_2_P4}
\end{align}
Note that, the definition above automatically satisfies the constraints $({\rm P}5a)$, $({\rm P}5b)$ and $({\rm P}5c)$ in P5.
Then, by always using the definition in~\eqref{eq:P3_2_P4}, we can equivalently rewrite the constraint $({\rm P}3b)$ as 
\[
    ({\rm P}3b):\ \sum_{\substack{p \in \mcal{P},\\(i,j) \in p,\\ j = p\pnext(i)}} \lambda_{\ell_{j,i}}^p \leq \lambda_{\ell_{j,i}} ,\qquad \forall (j,i) \in [1:N+1]\times[0:N].
\]
Now, if we fix $\hat{i} \in [0:N]$ and add the left-hand side and right-hand side of $({\rm P}3b)$ for $(j,i) \in [1:N+1]\times \{\hat{i}\}$, then we get
\begin{align*}
    \forall \hat{i} \in [0:N],\qquad  &\sum_{j \in [1:N+1]}\sum_{\substack{p \in \mcal{P},\\(\hat{i},j )\in p,\\ j = p\pnext(\hat{i})}} \lambda_{\ell_{j,\hat{i}}}^p \leq \sum_{j \in [1:N+1]}\lambda_{\ell_{j,\hat{i}}}\\
    &\implies \sum_{\substack{p \in \mcal{P}_{\hat{i}}}}  \lambda_{\ell_{p\pnext(\hat{i}),\hat{i}}}^p \leq \sum_{j \in [1:N+1]}\lambda_{\ell_{j,\hat{i}}} \stackrel{({\rm P}3c)}{\leq 1} \implies ({\rm P}5d).
\end{align*}
Similarly, by adding the constraints in $({\rm P}3b)$ for a fixed $\hat{j} \in [1:N+1]$, one can show that, under the transformation in~\eqref{eq:P3_2_P4}, the constraint in $({\rm P}5e)$ is satisfied.
Thus, for any feasible point in P3, we can get a feasible point in P5 using the transformation in~\eqref{eq:P3_2_P4}. Regarding the objective function, note that we did not perform any transformation on the variables $F_p$ from P3 to P5. 
It therefore follows that the objective function value achieved in P3 is the same as the one achieved in P5.

\noindent \underline{P5 $\to$ P3.} 
Given a feasible point in P5, we define the following variables for each link in the network
\[
    \lambda_{\ell_{j,i}} = \sum_{\substack{p \in \mcal{P},\\(i,j) \in p,\\j = p\pnext(i)}} \lambda^p_{\ell_{j,i}}.
\]
Based on this transformation, we automatically have that $({\rm P}3e)$ is satisfied. Moreover, we have that
\begin{align*}
({\rm P}5a)\ &: \quad \forall p \in \mcal{P}, \quad 0 \leq F_p  &\implies  ({\rm P}3a)\\
({\rm P}5d)\ &: \quad \forall i,\quad 1 \geq \sum_{p \in \mcal{P}_i} \lambda^p_{\ell_{p\pnext(i),i}} = \sum_{\substack{p \in \mcal{P},\\ (i,j) \in p}} \lambda_{\ell_{j,i}} = \sum_{j \in [1:N+1]\backslash\{i\}} \lambda_{\ell_{j,i}} &\implies ({\rm P}3c)\\
({\rm P}5e)\ &: \quad \forall i,\quad 1 \geq \sum_{p \in \mcal{P}_i} \lambda^p_{\ell_{i,p\pprev(i)}} = \sum_{\substack{p \in \mcal{P},\\ (j,i) \in p}} \lambda_{\ell_{i,j}} = \sum_{j \in [0:N]\backslash\{i\}} \lambda_{\ell_{i,j}} &\implies ({\rm P}3d)\\
({\rm P}5b)\& ({\rm P}5c)\ &: \quad \sum_{\substack{p \in \mcal{P},\\ (i,j) \in p,\\ j = p\pnext(i)}} \frac{F_p}{\ell_{j,i}} = \sum_{\substack{p \in \mcal{P},\\ (i,j) \in p,\\ j = p\pnext(i)}} \lambda^p_{\ell_{j,i}} = \lambda_{\ell_{j,i}} &\implies ({\rm P}3b).
        \end{align*}
        Furthermore, note that the objective function in P5 and P3 is the same. Thus, a feasible point in P5 can be mapped to a feasible point in P3 with the same objective function value.
        In conclusion, the problems P2, P3 and P5 are equivalent. We now show that P5 is equivalent to P1 in Theorem~\ref{thm:MainTh}.

\noindent \underline{P5 $\to$ P1.} Define $x_p$ to be
        \begin{align}
            x_p = \frac{F_p}{\mathsf{C}_p} ,\quad \forall p \in \mcal{P}.
            \label{eq:P4_2_P1}
        \end{align}
Using this transformation, we get that the constraints in P4 imply the following
        \begin{align*}
({\rm P}5a)\ &: \quad \forall p \in \mcal{P}, \quad 0 \leq F_p = x_p \mathsf{C}_p  &\implies  ({\rm P}1a)\\
            ({\rm P}5d)\ &: \quad \forall i \in [0:N], \quad 1 \geq \sum_{p \in \mcal{P}_i} \lambda^p_{\ell_{p\pnext(i),i}} \stackrel{({\rm P}5b)}{=} \sum_{p \in \mcal{P}_i} \frac{F_p}{\ell_{p\pnext(i),i}} = \sum_{p \in \mcal{P}_i} \frac{x_p \mathsf{C}_p}{\ell_{p\pnext(i),i}} \stackrel{\eqref{eq:actTime}} {=} \sum_{p \in \mcal{P}_i} x_p f^p_{p\pnext(i),i}  & \implies ({\rm P}1b)\\
            ({\rm P}5e)\ &: \quad \forall i \in [1:N{+}1], \quad 1 \geq \sum_{p \in \mcal{P}_i} \lambda^p_{\ell_{i,p\pprev(i)}} \stackrel{({\rm P}5c)}{=} \sum_{p \in \mcal{P}_i} \frac{F_p}{\ell_{i,p\pprev(i)}} = \sum_{p \in \mcal{P}_i} \frac{x_p \mathsf{C}_p}{\ell_{i,p\pprev(i)}} \stackrel{\eqref{eq:actTime}} {=} \sum_{p \in \mcal{P}_i} x_p f^p_{i,p\pprev(i)}  & \implies ({\rm P}1c).    
        \end{align*}
Moreover, we have that
        \begin{align*}
            (\text{P5 objective function})\ &: \sum_{p \in \mcal{P}} F_p=  \sum_{p \in \mcal{P}} x_p \mathsf{C}_p = (\text{P1 objective function}).
        \end{align*}
Thus, for any feasible point in P5, we get a feasible point in P1 using the transformation in~\eqref{eq:P4_2_P1} that has the objective function with the same value as the original point in P5.

\noindent
\underline{P1 $\to$ P5.} Define $F_p$, $\lambda^p_{\ell_{p\pnext(i),i}}$ and $ \lambda^p_{\ell_{i, p\pprev(i)}}$ as
        \begin{align}
F_p = x_p \mathsf{C}_p, \quad \lambda^p_{\ell_{p\pnext(i),i}} = \frac{x_p \mathsf{C}_p}{\ell_{p\pnext(i),i}} \forall i \in p \backslash\{N+1\}, \quad \lambda^p_{\ell_{i, p\pprev(i)}} = \frac{x_p \mathsf{C}_p}{\ell_{i, p\pprev(i)}} \forall i \in p \backslash\{0\}
            \label{eq:P1_2_P4}
        \end{align}
that hold $\forall p \in \mcal{P}$.
        Note that the transformation above directly implies conditions $({\rm P}5b)$ and $({\rm P}5c)$ in P5. Now, we are going to show that the constraints in P1 when applied to~\eqref{eq:P1_2_P4} imply the rest of the constraints in P5 as follows
        \begin{align*}
            ({\rm P}1a)\ &: \forall p \in \mcal{P}, \quad 0 \leq x_p = \frac{F_p}{\mathsf{C}_p} &\implies ({\rm P}5a)\\
            ({\rm P}1b)\ &: \forall i \in [0:N] \quad 1 \geq \sum_{p \in \mcal{P}_i}  x_p f^p_{p\pnext(i),i}  \stackrel{\eqref{eq:actTime}} {=} \sum_{p \in \mcal{P}_i} \frac{x_p \mathsf{C}_p}{\ell_{p\pnext(i),i}} = \sum_{p \in \mcal{P}_i} \lambda^p_{\ell_{p\pnext(i),i}}   &\implies ({\rm P}5d)\\
            ({\rm P}1c)\ &: \forall i \in [1:N+1] \quad 1 \geq \sum_{p \in \mcal{P}_i}  x_p f^p_{i,p\pprev(i)}  \stackrel{\eqref{eq:actTime}} {=} \sum_{p \in \mcal{P}_i} \frac{x_p \mathsf{C}_p}{\ell_{i,p\pprev(i)}} = \sum_{p \in \mcal{P}_i} \lambda^p_{\ell_{i,p\pprev(i)}}   &\implies ({\rm P}5e)
        \end{align*}
Moreover, we have that
        \begin{align*}
            (\text{P1 objective function})\ &: \sum_{p \in \mcal{P}} x_p \mathsf{C}_p =\sum_{p \in \mcal{P}} F_p = (\text{P5 objective function}).
        \end{align*}
Thus, for any feasible point in P1, we get a feasible point in P5 using the transformation in~\eqref{eq:P1_2_P4} that has the objective function with the same value as the original point in P1.
Thus, the two problems P1 and P5 are equivalent.
In conclusion, the problems P1, P2, P3 and P5 are equivalent. This concludes the proof of Theorem~\ref{thm:MainTh}.
{
\section{A polynomial algorithm to compute the optimal schedule for Gaussian FD 1-2-1 networks}\label{app:poly_algorithm}
In this appendix, we show that for the approximate capacity expression in~\eqref{eq:apprCap}, we can compute the optimal schedume $\lambda^\star$ as well as the value of $\msf{C}_{\rm cs,iid}$ in polynomial time. 

To start of, we note - as in Section~\ref{sec:ProofMainTh} - that for a fixed $\lambda_s$, the inner minimization in~\eqref{eq:apprCap} is the standard min-cut problem over a graph with link capacities given by 
\begin{align}
\label{eq:ellMaxFlow2}
\ell_{j,i}^{(s)} = \left( \sum_{\substack{ s:\\ j \in s_{i,t},\ i \in s_{j,r} }} \lambda_s \right) \ell_{j,i}.
\end{align}
Thus, we can replace the inner minimization in~\eqref{eq:apprCap} with the max-flow problem over the graph with link capacities defined in~\eqref{eq:ellMaxFlow2} to give the linear program $\rm Pflow_1$, i.e., 
\begin{align*}
\begin{array}{llll}
&{\rm Pflow_1}&{\rm\ :} \ \msf{C}_{\rm cs,iid} =\max \displaystyle\sum_{j =1}^{N+1} F_{j,0}&\nonumber\\
& ({\rm Pf}1a)\ & 0 \leq F_{j,i} \leq  \left( \displaystyle\sum_{\substack{ s:\\ j \in s_{i,t},\ i \in s_{j,r} }} \lambda_s \right) \ell_{j,i} & (i,j) \in [0:N]\times[1:N{+}1], \\
& ({\rm Pf}1b)\ & \displaystyle\sum_{j\in[1:N{+}1]\backslash\{i\}} F_{j,i} = \displaystyle\sum_{k\in[0:N]\backslash\{i\}} F_{i,k} & i \in [1:N], \\
& ({\rm Pf}1c)\ &\displaystyle\sum_s \lambda_s \leq 1, \\
& ({\rm P}1d)\ &\lambda_s \geq 0 & \forall s,
\end{array}
\end{align*}

where $F_{j,i}$ is the flow through the link going from node $i$ to node $j$ and $\lambda_s$ is a state of the 1-2-1 network. A solution to the LP $\rm Pflow_1$ gives us the value $\msf{C}_{\rm cs,iid}$ as well as the optimal schedule to achieve the approximate capacity. Unfortunately, $\rm Pflow_1$ has an exponential number of variables $\lambda_s$ and therefore, cannot be solved efficiently in its current form.

Our main goal is to show that $\rm Pflow_1$ can be equivalently written as the LP $Pflow_2$ below. 
\begin{align*}
{\rm Pflow_2}&{\rm\ :} \ \msf{C}_{\rm cs,iid} = \max \sum_{j =1}^{N+1} F_{j,0}& \nonumber\\
& 0 \leq F_{j,i} \leq \lambda_{\ell_{j,i}} \ell_{j,i} & (i,j) \in [0:N]\times[1:N{+}1], \\
& \sum_{j\in[1:N{+}1]\backslash\{i\}} F_{j,i} = \sum_{k\in[0:N]\backslash\{i\}} F_{i,k} & i \in [1:N], \\
& \displaystyle \sum_{\substack{j \in [1:N{+}1]\backslash\{i\}}} \! \! \! \!  \lambda_{\ell_{j,i}} \leq 1 & \forall i \in [0:N], & \\
& \displaystyle \sum_{\substack{i \in [0:N]\backslash\{j\}}} \! \! \! \!  \lambda_{\ell_{j,i}} \leq 1 & \forall j \in [1:N+1], & \\
&\lambda_{\ell_{j,i}} \geq 0  &\forall (i,j)  {\in} [0\!:\!N] \!\times\! [1\!:\!N{+}1], &
\end{align*}

where $\lambda_{\ell_{j,i}}$ presents the fraction of time during which the links $i \to j$ is active. 

Assuming this is true, then we have the following appealing outcomes:
\begin{enumerate}[(a)]
    \item Since $\rm Pflow_2$ has a polynomial number of variables and constraints in $N$, then we can compute the value of $\msf{C}_{\rm cs,iid}$ in polynomial time in $N$.
    \item If the mapping from an optimal point in $\rm Pflow_2$ to an optimal point in $Pflow_1$ can be done in polynomial time, then we have an algorithm to find the optimal schedule of the Gaussian FD 1-2-1 network in polynomial time. This can be done by first solving $\rm Pflow_2$ in polynomial time and then mapping its optimal solution in polynomial time to an optimal schedule in $\rm Pflow_1$.
\end{enumerate}

In what follows, we show that the $\rm Pflow_1$ and $\rm Pflow_2$ are indeed equivalent and the mapping an optimal point in $\rm Pflow_2$ to $\rm Pflow_1$ can be done by a construction that is polynomial in $N$. Note that in $\rm Pflow_1$ and $\rm Pflow_2$, the variables $F_{j,i}$ are the same, therefore we only need to find the mapping between $\{\lambda_s\}$ and $\{\lambda_{\ell_{j,i}}\}$.

\noindent \underline{$\rm Pflow_1 \to Pflow_2$.} Given a feasible point in $\rm Pflow_1$ we define
\[
    \lambda_{\ell_{j,i}} = \sum_{\substack{ s:\\ j \in s_{i,t},\ i \in s_{j,r} }} \lambda_s
\]

Using this definition, we have that 
        \begin{align*}
            ({\rm Pf}1a)\ &: \forall (i,j) \quad   F_{j,i} \leq  \left( \displaystyle\sum_{\substack{ s:\\ j \in s_{i,t},\ i \in s_{j,r} }} \lambda_s \right) \ell_{j,i} = \lambda_{\ell_{j,i}} \ell_{j,i} &\implies ({\rm Pf}2a)\\
            ({\rm Pf}1c)\ &: \forall i \in [0:N] \quad \sum_{j=1}^{N+1} \sum_{\substack{ s:\\ j \in s_{i,t},\ i \in s_{j,r} }}  \leq \sum_s \lambda_s  \leq 1   &\implies ({\rm Pf}2c)\\
            ({\rm Pf}1c)\ &: \forall j \in [1:N+1] \quad \sum_{i=0}^{N} \sum_{\substack{ s:\\ j \in s_{i,t},\ i \in s_{j,r} }}  \leq \sum_s \lambda_s  \leq 1   &\implies ({\rm Pf}2d).
        \end{align*}
        In addition, since the variables $F_{j,i}$ are not changed in the mapping then the new mapped point in $\rm Pflow_2$ has the save objective value as the original point in $\rm Pflow_1$.

\noindent \underline{$\rm Pflow_2 \to Pflow_1$.} Given a feasible point in $\rm Pflow_2$ we would like to construct a set of $\lambda_s$ that represent states in the FD 1-2-1 network, which collectively activate each link $(i,j)$ for at least the fraction dictated by $\lambda_{\ell_{j,i}}$.

To map $\rm Pflow_2$ to $\rm Pflow_1$, we use the same visualization introduced in Section~\ref{sec:ProofMainTh}. In particular, we divide each node $i \in [0:N+1]$ in the network into two vertices ($i_T$ and $i_R$) representing the transmitting and receiving functions of the node;
note that $0_R = (N+1)_T = \emptyset$ since the source (node $0$) is always transmitting and the destination (node $N+1$) is always receiving.
This gives us the bipartite graph $\mcal{G}_B = (\mcal{T},\mcal{R},\mcal{E})$, where the vertices $\mcal{T}$ (respectively, $\mcal{R}$) are the transmitting modules of our nodes (respectively, $\mcal{R}$ collects our receiving modules), and we have an edge  $(i_T,j_R) \in \mcal{E}$ for each link  in the network. 
It is easy to see that a valid state in $\rm Pflow_1$ represents a matching in the bipartite graph $\mcal{G}_B$. A perfect matching in a bipartite graph is represented by a permutation matrix $P$ where the rows of the matrix represent the set of vertices $\mcal{R}$ and the columns are indexed by the vertices in $\mcal{T}$. 
Furthermore, we can write the feasible point in $\rm Pflow_2$ as a weighted adjacency matrix of the graph $\mcal{G}_B$. In particular, $\lambda_{\ell_{j,i}}$ represents the weight of the edge connecting vertex $i_T$ to vertex $j_R$.

At this point, we can explicitly express our desired mapping in terms of the bipartite graph $\mcal{G}_B$: Given a weighted adjacency matrix $L$ (which is filled using a feasible point of $\rm Pflow_2$ as $[L]_{ji} = \lambda_{\ell_{j,i}}$), can we efficiently find a set of permutation matrices $\{P_i\}$ that satisfy
\begin{align}
\label{eq:birkoff-vonneuman}
 L \leq \sum_{i=1}^K \varphi_i P_i,\quad \sum_{i}^K \varphi_i = 1, \varphi \geq 0. 
\end{align}
In particular, we are interested in a polynomial time approach to find these $P_i$ matrices. To answer this question, we need to observe some interesting property of $L$.
\begin{align*}
 \forall (i,j) \in [0:N+1]^2, [L]_{ji} \geq 0,\\
 \forall i \in [0:N+1], \sum_{j=0}^{N+1} [L]_{ji} \leq 1,\\
 \forall j \in [0:N+1], \sum_{i=0}^{N+1} [L]_{ji} \leq 1.
\end{align*}
Such a matrix $L$ is called a \emph{doubly-substochastic matrix}. The result in~\cite{chang1999service} provides an algorithm that finds a set of permutation matrices satisfying~\eqref{eq:birkoff-vonneuman} for any doubly sub-stochastic matrix in $\mbb{R}^{N\times N}$. The algorithm runs in $O(N^{4.5})$ time and outputs $N^2 - 2N +2$ permutation matrices. This proves the existence of a mapping from $\rm Pflow_2$ to $\rm Pflow_1$ that can be done in polynomial time.
}

\bibliographystyle{IEEEtran}
\bibliography{mmWaveNetwork}
\end{document}